\documentclass[journal,onecolumn]{IEEEtran}
\usepackage{float} 
\usepackage{algorithmic}
\usepackage{graphicx}
\usepackage{textcomp}
\usepackage{cite}
\usepackage{amsthm,amsmath,amssymb}
\usepackage{mathrsfs}
\usepackage{multirow}
\usepackage{microtype}
\usepackage{hyperref}
\usepackage[capitalize]{cleveref}
\usepackage{verbatim}
\usepackage{xcolor}
\usepackage{enumitem}
\usepackage{arydshln} % for dashlines in tables/arrays

\newtheorem{theorem}{Theorem}[section]
\newtheorem{lemma}[theorem]{Lemma}

\newtheorem{definition}[theorem]{Definition}

\newtheorem{remark}[theorem]{Remark}
\newtheorem{claim}[theorem]{Claim}

\begin{document}

\title{The Capacity of  Private Information Retrieval from Coded Storage Systems with Colluding, Robust and Byzantine Servers}
\author{{\textbf{Rui Sun} and \textbf{Yiwei Zhang}}
\thanks{The authors are with State Key Laboratory of Cryptography and Digital Economy Security, Key Laboratory of Cryptologic Technology and Information Security of Ministry of Education, School of Cyber Science and Technology, Shandong University, Qingdao, Shandong, 266237, China. }
\thanks{Research supported by National Key Research and Development Program of China under Grant No. 2022YFA1004900, National Natural Science Foundation of China under Grant No. 12231014, and the Taishan Scholars Program. 
Corresponding author: Y. Zhang (ywzhang@sdu.edu.cn).}
}

\maketitle

\begin{abstract}
We consider the problem of private information retrieval 
from MDS-coded databases with $T$ colluding servers,
$U$ robust  servers, and $B$ Byzantine servers,
referred to as U-B-MDS-TPIR. In the U-B-MDS-TPIR setting, $M$ files are stored across $N$ servers, and in this paper each file is independently encoded using an $(N,K)$ GRS code. We characterize the  capacity
among linear, full support-rank U-B-MDS-TPIR schemes as
$C=\frac{N-U-2B}{N-U}(1+\frac{K+T-1}{N-U-2B}+\cdots+(\frac{K+T-1}{N-U-2B})^{M-1})^{-1}$
when $N-U-2B\ge KT$.
When $K=1$ or $T=1$, the same expression gives the capacity
without the linearity and full support-rank restrictions.  We further extend the construction to the retrieval of \(P\) out of \(M\) files. For \(P \ge \frac{M}{2}\), the resulting multi-file scheme achieves capacity when \(K=1\) or \(T=1\).
\end{abstract}

\section{Introduction}
Private information retrieval (PIR) enables a user to retrieve a desired data item from a collection of servers without revealing which item is being requested. The problem was introduced by Chor et al.~\cite{492461,chor1998private}, where the objective is to privately retrieve one bit from a database while minimizing the total communication between the user and the servers, including both the upload of the queries by the user and the download of the responses from servers. In the information-theoretic formulation of PIR for distributed storage systems, the data are typically organized into files whose sizes can be arbitrarily large. In this regime, the upload cost of the queries is usually negligible compared with the download cost~\cite{chan2015private}, and the efficiency of a PIR scheme is  measured primarily by its download rate, defined as the ratio between the size of the desired file and the total amount of downloaded information. For a comprehensive overview of PIR and related topics, see the survey articles~\cite{gasarch2004survey,ulukus2022private,vithana2023private,d2024guided} and the references therein.

 In the MDS-TPIR setting, $M$ files are stored across $N$ servers, with each file independently encoded by the same $[N,K]$ MDS code. A user wishes to retrieve one of the files without revealing its index to any set of up to $T$ colluding servers. The {\it PIR rate} is defined as the ratio between the size of the desired file and the total amount of downloaded information. The maximum achievable rate is called the {\it MDS-TPIR capacity}.
Sun and Jafar~\cite{sun_12017capacity} characterized the capacity of classical PIR with replicated storage and non-colluding servers as $C_{\mathrm{PIR}}=(1+\frac{1}{N}+\cdots+\frac{1}{N^{M-1}})^{-1}$.  For the non-colluding case $T=1$, Banawan and Ulukus~\cite{banawan2018capacity} characterized the exact MDS-PIR capacity as $C_{\mathrm{MDS\text{-}PIR}}=(1+\frac{K}{N}+\cdots+(\frac{K}{N})^{M-1})^{-1}$. For the repetition-coded case $K=1$, Sun and Jafar~\cite{sun_22017capacity} characterized the TPIR capacity as $C_{\mathrm{TPIR}}=(1+\frac{T}{N}+\cdots+(\frac{T}{N})^{M-1})^{-1}$. Thus, for the degenerate cases $K=1$ or $T=1$, the MDS-TPIR capacity is completely characterized by $(1+\rho+\cdots+\rho^{M-1})^{-1}$, where $\rho=\frac{T}{N}$ for $K=1$ and $\rho=\frac{K}{N}$ for $T=1$.
For the non-degenerate case with $K\geq 2$ and $T\geq 2$, several achievable schemes have been proposed~\cite{freij2017private,zhang2019general,d2019one,holzbaur2021toward,sun2026refinement,sun2017private,sun2025mds,sun2026disguise}. In particular, for  $N\ge KT$, the capacity within
the class of linear, full support-rank MDS-TPIR schemes is given
by the same expression $(1+\rho+\cdots+\rho^{M-1})^{-1}$,
with $\rho=\frac{K+T-1}{N}$~\cite{d2019one,holzbaur2021toward,sun2026refinement}.
However, without the linearity and full support-rank restrictions,
the exact MDS-TPIR capacity remains unknown in general.

In addition to preserving user privacy, PIR systems may also need to tolerate server failures and malicious behavior. Two common models are robust (unresponsive) servers, which return no answer, and Byzantine servers, which may return arbitrary erroneous answers. For replicated storage, Sun and Jafar~\cite{sun_22017capacity} characterized the capacity of PIR with colluding and robust servers. Banawan and Ulukus~\cite{Byzantinebanawan2018capacity} subsequently characterized the capacity of PIR from replicated databases with  colluding, robust and Byzantine servers.
Robust PIR from coded storage was studied in~\cite{tajeddine2017robust}. Zhang and Ge~\cite{zhang2017private} further extended the MDS-TPIR model in two separate directions by considering robust servers and Byzantine servers, respectively. Tajeddine et al.~\cite{tajeddine2019private} proposed an explicit scheme for Reed--Solomon coded storage with $T$ colluding, $U$ robust and $B$ Byzantine  servers. When $N-U-2B\geq K+T$, the scheme achieves the rate
\begin{align}
R
=
1-\frac{K+T+2B-1}{N-U}.
\label{eq:Taj-rate}
\end{align}
They further conjectured that \eqref{eq:Taj-rate} gives the asymptotic capacity as $M\to\infty$. This conjecture is known to hold for $K=1$ \cite{Byzantinebanawan2018capacity}. Other extensions have considered multi-round Byzantine PIR~\cite{yao2019capacity} and $X$-secure U-B-MDS-TPIR~\cite{jia2020xsecure}.   However, except for the case $K=1$, the exact finite-$M$ capacity of U-B-MDS-TPIR remains unknown. In particular, even the non-colluding case $T=1$ has not been characterized.

Another related extension is multi-file PIR, where a user wishes to retrieve $P$ files simultaneously. Banawan and Ulukus~\cite{banawan2018multi} studied multi-file PIR for replicated storage and showed that joint retrieval can be more efficient than repeatedly applying a single-file PIR scheme. They characterized the exact capacity for $K=T=1$ when $P\geq M/2$. Subsequently, Zhang and Ge~\cite{zhang2017private} extended the problem to MDS-coded storage with colluding servers and proposed a multi-file MDS-TPIR scheme for $P\geq M/2$, which achieves the capacity when either $K=1$ or $T=1$. 

In this paper, we consider PIR from an $[N,K]$ GRS-coded storage system with $M$ files, $T$ colluding servers, $U$ robust servers, and $B$ Byzantine servers, which we refer to as the U-B-MDS-TPIR problem. Our main contributions are summarized as follows.
\begin{itemize}
    \item We characterize the  capacity
among linear, full support-rank  U-B-MDS-TPIR schemes as 
    \begin{align}
    C
    =
    \frac{N-U-2B}{N-U}
    \left(
    1+\frac{K+T-1}{N-U-2B}
    +\cdots+
    \left(\frac{K+T-1}{N-U-2B}\right)^{M-1}
    \right)^{-1},
    \label{eq:intro-general-rate}
    \end{align}
    for $N-U-2B\geq KT$. When $K=1$ or $T=1$, the same expression gives the capacity
without the linearity and full support-rank restrictions.
    \item  We extend the proposed U-B-MDS-TPIR scheme to the multi-file setting.
For $N-U-2B\geq KT$, the resulting
$P$-out-of-$M$ multi-file U-B-MDS-TPIR scheme achieves the rate
\[
R=
\frac{N-U-2B}{N-U}
\left(
1+\frac{(K+T-1)(M-P)}{P(N-U-2B)}
\right)^{-1}.
\]
For $P \ge \frac{M}{2}$, we further show that this rate is
capacity-achieving when either $K=1$ or $T=1$.
\end{itemize}

The  capacity expression in \eqref{eq:intro-general-rate} also reveals the connection between the finite-$M$ and asymptotic results. As $M\rightarrow\infty$,
\begin{align*}
\frac{N-U-2B}{N-U}
    \left(
    1+\frac{K+T-1}{N-U-2B}
    +\cdots+
    \left(\frac{K+T-1}{N-U-2B}\right)^{M-1}
    \right)^{-1}
\longrightarrow
1-\frac{K+T+2B-1}{N-U},
\end{align*}
which coincides with the previously conjectured asymptotic capacity in \eqref{eq:Taj-rate}. Moreover, the rate in \eqref{eq:intro-general-rate} decreases with $M$ and  approaches  this asymptotic value, showing an additional finite-$M$ gain for every finite $M$.

\section{Preliminaries}\label{Sec:pre}
For positive integers $a \le b$, let $[a:b] = \{a, a+1, \dots, b\}$, and abbreviate $[1:n]$ as $[n]$. Given an index set $\mathcal{I} = \{i_1, \dots, i_k\} \subseteq [N]$ with $i_1 < \dots < i_k$ and a vector $\mathbf{x} = (x_1, \dots, x_N)$, we write $\mathbf{x}_{\mathcal{I}} = (x_{i_1}, \dots, x_{i_k})$. The inner product of two vectors is denoted by $\langle \mathbf{x}, \mathbf{y} \rangle$.
Let $\mathbf{G}$ be a $K \times N$ matrix over $\mathbb{F}_q$. We denote by $\operatorname{rowsp}(\mathbf{G})$ the linear code generated by its rows, and by $\mathbf{G}(:,\mathcal{I})$ the submatrix formed by its columns indexed by $\mathcal{I}$. For a positive integer $a$, define $
\mathbf{G}_{\langle a \rangle} \triangleq [\mathbf{G} \mid \mathbf{G} \mid \cdots \mid \mathbf{G}] \in \mathbb{F}_q^{K \times aN},$
where $\mathbf{G}$ is concatenated $a$ times horizontally.
Finally, we write $X \sim Y$ to indicate that $X$ and $Y$ are identically distributed.

\subsection{Problem setting of U-B-MDS-TPIR}

Consider $M$ independent files $W^{[1]}, \dots, W^{[M]} \in \mathbb{F}_q^{\frac{L}{K}\times K}$. In $q$-ary units,
\begin{align*}
    H(W^{[1]}) &= \dots = H(W^{[M]}) = L,  \\
    H(W^{[1]}, \dots, W^{[M]}) &= H(W^{[1]}) + \dots + H(W^{[M]}). 
\end{align*}

The files are stored in a system consisting of $N$ servers. The  server $n$ stores $(W^{[1]}_{n}, W^{[2]}_{n}, \dots, W^{[M]}_{n})$, where $W^{[m]}_{n} \in \mathbb{F}_q^{\frac{L}{K} \times 1}$ represents $\frac{L}{K}$ coded symbols of $W^{[m]}$, $m \in [M]$, $n\in[N]$. The storage system applies an $[N,K]$-MDS code independently to each file. Thus the following hold:
\begin{align*}
    H(W^{[m]}_{n} | W^{[m]}) &= 0, \quad H(W^{[m]}_{n}) = L/K. \\
 [\mbox{MDS}] ~    H(W^{[m]} | W^{[m]}_{\mathcal{K}}) &= 0, \quad \forall \mathcal{K}\subseteq [N], \; |\mathcal{K}| = K. 
\end{align*}

To  privately retrieve $W^{[m]}$, $m\in [M]$, the user privately generates $N$ queries $Q_1^{[m]}, \dots, Q_N^{[m]}$. The queries are independent of the files, i.e.,
\begin{align*}
I(Q_1^{[m]},\dots,Q_N^{[m]};W^{[1]}, \dots, W^{[M]})=0.
\end{align*}

The user sends query $Q_n^{[m]}$ to server $n$ for each
$n\in[N]$. Let $A_n^{[m]}$ denote the truthful answer
prescribed by the protocol for server $n$, which is a
deterministic function of $Q_n^{[m]}$ and its stored
information $(W_n^{[1]},\ldots,W_n^{[M]})$. Thus,
\[
H\bigl(A_n^{[m]}
\mid Q_n^{[m]},W_n^{[1]},\ldots,W_n^{[M]}\bigr)=0,
\qquad \forall n\in[N].
\]

There exists a set $\mathcal{U}\subseteq[N]$ of $U$ robust
(unresponsive) servers that return no answer to the user.
For each $n\in[N]\setminus\mathcal{U}$, let
$\bar{A}_n^{[m]}$ denote the answer actually received by
the user from server $n$.
There also exists a set
$\mathcal{B}\subseteq[N]\setminus\mathcal{U}$ of $B$
Byzantine servers, whose identities are unknown to the user.
For each $n\in\mathcal{B}$, server $n$ may replace its
truthful answer $A_n^{[m]}$ with an arbitrary sequence
of the same prescribed length and send this sequence
to the user as $\bar{A}_n^{[m]}$. Each remaining server $n\in[N]\setminus(\mathcal{U}\cup\mathcal{B})$
is trustworthy and returns the truthful answer prescribed
by the protocol, i.e., $\bar{A}_n^{[m]}=A_n^{[m]}$.

 The user can decode the desired file $W^{[m]}$ from the responses that he receives,
\begin{align*}
 [\mbox{Correctness}] ~    H(W^{[m]} | \bar{A}_{[N]\setminus \mathcal{U}}^{[m]},Q_{1:N}^{[m]}) &= 0, \quad \forall   \mathcal{U}, \;\mathcal{B} \subseteq [N],\; |\mathcal{U}|=U, \;|\mathcal{B}|=B, \;\mathcal{U}\cap \mathcal{B}=\emptyset.
\end{align*}

To satisfy the T-privacy requirement, it must hold that
\begin{align*}
 [\mbox{Privacy}] ~    (Q_{\mathcal{T}}^{[1]},W_{\mathcal{T}}^{[1]},\dots,W_{\mathcal{T}}^{[M]})\sim (Q_{\mathcal{T}}^{[m]},W_{\mathcal{T}}^{[1]},\dots,W_{\mathcal{T}}^{[M]}), \quad \forall m\in [M],\; \forall \mathcal{T}\subseteq [N],\; |\mathcal{T}|=T.
\end{align*}

The retrieval rate is defined
as the ratio of the file size to the worst-case download cost
over all unresponsive sets of size $U$, 
\[
R=\frac{L}{D}
=\frac{L}{
\displaystyle
\max_{\substack{\mathcal{U}\subseteq[N]\\|\mathcal{U}|=U}}
\sum_{n\in[N]\setminus\mathcal{U}}
H(A_n^{[m]})
}.
\]
The PIR capacity is the maximum achievable rate among all available schemes. 

\begin{definition}[Linear,  Full Support-Rank PIR\cite{holzbaur2021toward}]
\label{def:linear-full-support-rank}
A PIR scheme is called linear if the answers  are obtained by taking inner products between query vectors and stored message vectors.
A linear PIR scheme is of full support-rank
if for any set of at most $T$  colluding servers, all nonzero query coefficient vectors they receive  with respect to any single file, are linearly independent. 
\end{definition}

\subsection{Preparations} 
\begin{definition}[GRS codes]
Let $\mathbb{F}_q$ be a finite field. For integers $N$ and $K$ with $1 \le K \le N \le q$, let $\boldsymbol{\alpha} = (\alpha_1, \dots, \alpha_N) \in \mathbb{F}_q^N$ consist of $N$ distinct elements and let $\mathbf{v} = (v_1, \dots, v_N) \in (\mathbb{F}_q^*)^N$ be a vector of $N$ nonzero multipliers (not necessarily distinct). The generalized Reed–Solomon code $\mathrm{GRS}_K(\boldsymbol{\alpha}, \mathbf{v})$ is defined as
\begin{align*}
\mathrm{GRS}_K(\boldsymbol{\alpha}, \mathbf{v})
&= \bigl\{\, \bigl(v_1 f(\alpha_1),\, v_2 f(\alpha_2),\, \dots,\, v_N f(\alpha_N)\bigr) 
\mid f \in \mathbb{F}_q[x],\ \deg f < K \,\bigr\}.
\end{align*}
It is an $[N, K, N-K+1]_q$ MDS code.
\end{definition}

\begin{definition}[Schur Product of Codes]
Let $C, C' \subseteq \mathbb{F}_q^N$ be two linear codes. The Schur product of $C$ and $C'$, denoted $C \star C'$, is the span of all component-wise products of codewords from $C$ and $C'$, i.e.,
\[
C \star C' = \operatorname{span}_{\mathbb{F}_q}\{\, \mathbf{c} \star \mathbf{c}' \mid \mathbf{c} \in C,\ \mathbf{c}' \in C' \,\},
\]
where for $\mathbf{c} = (c_{1}, \dots, c_{N})$ and $\mathbf{c}' = (c_{1}', \dots, c_{N}')$, their component-wise product is given by $\mathbf{c} \star \mathbf{c}' = (c_{1}c_{1}',\dots,c_{N}c_{N}')$. 
\end{definition}

\begin{definition}[Schur-product matrix]\label{def:matrix-schur}
Let $G \in \mathbb{F}_q^{K \times N}$ with rows $\mathbf{g}_1,\dots,\mathbf{g}_K$ and 
$G' \in \mathbb{F}_q^{K' \times N}$ with rows $\mathbf{g}_1',\dots,\mathbf{g}_{K'}'$. 
Their Schur product $G \star G' \in \mathbb{F}_q^{KK' \times N}$ is defined as
\[
G \star G' =
\begin{bmatrix}
\mathbf{g}_1 \star \mathbf{g}_1' \\
\vdots \\
\mathbf{g}_1 \star \mathbf{g}_{K'}' \\
\vdots \\
\mathbf{g}_K \star \mathbf{g}_1' \\
\vdots \\
\mathbf{g}_K \star \mathbf{g}_{K'}'
\end{bmatrix},
\]
where the rows are arranged in lexicographic order $(i,j)$ with $i \in [K], j \in [K']$, and $\star$ denotes component-wise multiplication. 
\end{definition}

\begin{claim}\label{claim}
Let $C, C' \subseteq \mathbb{F}_q^N$ be linear codes with generator matrices 
$G \in \mathbb{F}_q^{K \times N}$ and 
$G'\in \mathbb{F}_q^{K' \times N}$, respectively.
Let $W = (\mathbf{w}_1,\dots,\mathbf{w}_K) \in \mathbb{F}_q^{\ell \times K}$ 
and $V = (\mathbf{v}_1,\dots,\mathbf{v}_{K'}) \in \mathbb{F}_q^{\ell \times K'}$ 
be arbitrary matrices.
Define $\ell$-dimensional column vectors 
$\mathbf{x}_n, \mathbf{y}_n \in \mathbb{F}_q^\ell$ for $n \in [N]$ by
\[
(\mathbf{x}_1, \dots, \mathbf{x}_N)=W G,  \quad 
(\mathbf{y}_1, \dots, \mathbf{y}_N) = V G'.
\]
Then the following hold:
\begin{enumerate}
\item The Schur product code 
$C \star C'=\operatorname{rowsp}( G) \star \operatorname{rowsp}( G') $ 
is generated by the rows of the matrix $G \star G'$ defined in 
Definition~\ref{def:matrix-schur}. 
(Note that $G \star G'$ is not necessarily a generator matrix in the standard 
sense, as its row vectors may be linearly dependent.)
\item The vector of inner products satisfies
\[
\bigl(\langle \mathbf{x}_1, \mathbf{y}_1 \rangle,\;
\dots,\;
\langle \mathbf{x}_N, \mathbf{y}_N \rangle \bigr)
=
\bigl(\mathbf{w}_1^\top V,\;
\dots,\;
\mathbf{w}_K^\top V\bigr)
(G \star G').
\]
Consequently, it is a linear combination of the rows of $G \star G'$, 
and therefore lies in $C \star C'$.
\end{enumerate}
\end{claim}

\begin{IEEEproof}
The first statement is immediate from the definition of $G \star G'$ and the 
bilinearity of the Schur product: its rows are all $\mathbf{g}_k \star \mathbf{g}_{k'}'$, 
which span $C \star C'$.

For the second statement, for each $n\in[N]$, write
\[
\mathbf{x}_n
=
\sum_{i=1}^K G_{i,n}\mathbf{w}_i,
\qquad
\mathbf{y}_n
=
\sum_{j=1}^{K'} G'_{j,n}\mathbf{v}_j,
\]
where $G_{i,n}$ and $G'_{j,n}$ are the $(i,n)$-th and $(j,n)$-th entries 
of $G$ and $G'$, respectively. Then
\[
\langle\mathbf{x}_n,\mathbf{y}_n\rangle
=
\sum_{i=1}^K\sum_{j=1}^{K'}
\langle \mathbf{w}_i,\mathbf{v}_j\rangle
G_{i,n}G'_{j,n}.
\]
This is exactly the $n$-th entry of
$\bigl(\mathbf{w}_1^\top V,\dots,\mathbf{w}_K^\top V\bigr)
(G \star G')$. Hence the inner product vector is a linear combination of 
the rows of $G \star G'$. By the first statement, these rows generate 
$C \star C'$, so the inner product vector lies in $C \star C'$.
\end{IEEEproof}

\begin{definition}[Uniform Column Properties]
\label{def:uniform}
Let $D = [\,D_1 \;\big|\; D_2 \;\big|\; \cdots \;\big|\; D_a\,]$ be an $r\times aN$ matrix with each $D_i \in \mathbb{F}_q^{r\times N}$.
For an index set $\mathcal{L} \subseteq [N]$, denote its complement by $\mathcal{L}^c := [N]\setminus \mathcal{L}$, and define
\begin{align*}
D_{\mathcal{L}}& := [\,D_1(:,\mathcal{L}) \;\big|\; D_2(:,\mathcal{L}) \;\big|\; \cdots \;\big|\; D_a(:,\mathcal{L})\,],\\
D_{\mathcal{L}^c}& := [\,D_1(:,\mathcal{L}^c) \;\big|\; D_2(:,\mathcal{L}^c) \;\big|\; \cdots \;\big|\; D_a(:,\mathcal{L}^c)\,].
\end{align*}
For an integer $\ell$, we say that $D$ satisfies:
\begin{enumerate}
\item[(i)] the \emph{$(\ell,a)$-uniform selection full-rank property} if 
\[ \forall  \mathcal{L}\subseteq [N],\ |\mathcal{L}|=\ell, \quad
\operatorname{rank}(D_{\mathcal{L}}) = a \ell;
\]
\item[(ii)] the \emph{$(\ell,a)$-uniform puncturing MDS property} if 
\[\forall \mathcal{L}\subseteq [N],\ |\mathcal{L}|=\ell, \quad D_{\mathcal{L}^c} \text{ is column MDS}.\]
\end{enumerate}
\end{definition}

\begin{remark} Here, ``uniform'' means that the same set of column indices is selected or deleted from every block $D_i$.
Definition~\ref{def:uniform} states two properties. 
The first, the $(\ell,a)$-uniform selection full-rank property, means that uniformly selecting the same $\ell$ columns from every $N$-column block of $D$ yields an $r \times a\ell$ matrix of full column rank. 
The second, the $(\ell,a)$-uniform puncturing MDS property, means that uniformly puncturing the same $\ell$ columns from every $N$-column block of $D$ yields an $r \times a(N-\ell)$ matrix which is column-MDS. Note that $G\star G'$ and $G'\star G$ differ only by a row permutation. Since row permutations do not affect column properties, the two matrices satisfy the same uniform column properties.
\end{remark}

The following lemma restates, in transposed form,
\cite[Lemma~2]{yao2021capacity},
which generalizes the square-matrix result of Sun and Jafar
\cite[Lemma~1]{sun_22017capacity}
to full-rank rectangular matrices.

\begin{lemma}\label{lem:dim}
Let $S_1,S_2,\dots,S_M\in\mathbb{F}_q^{a\times a}$ be $M$ random
matrices, chosen independently and uniformly from all $a\times a$
full-rank matrices over $\mathbb{F}_q$.
Let $H_1,H_2,\dots,H_M\in\mathbb{F}_q^{b\times c}$ be $M$ fixed
matrices of full column rank, where $c\leq b\leq a$.
Let $\mathcal{I}_1,\mathcal{I}_2,\dots,\mathcal{I}_M$ be $M$
index vectors, each containing $b$ distinct indices from $[1:a]$.
Then
\begin{align*}
\left(
S_1(:,\mathcal{I}_1)H_1,\,
S_2(:,\mathcal{I}_2)H_2,\,
\dots,\,
S_M(:,\mathcal{I}_M)H_M
\right)
\sim
\left(
S_1(:,1:c),\,
S_2(:,1:c),\,
\dots,\,
S_M(:,1:c)
\right).
\end{align*}
\end{lemma}

\begin{lemma}[Error-correcting capability of MDS codes \cite{roth2006introduction}]
\label{lem:error}
Let $C$ be an $[N,K]$ MDS code over $\mathbb{F}_q$. Then $C$ can correct any error pattern of weight at most 
\[
\tau := \left\lfloor \frac{N-K}{2} \right\rfloor.
\]
Equivalently, the transmitted codeword can be uniquely recovered whenever the number of errors is at most $\tau$.
\end{lemma}

\section{The MDS-TPIR scheme with robust and Byzantine servers}\label{Sec:main}

The main result of this paper is summarized in the following theorem.

\begin{theorem}\label{theorem:U-B}
    When $N-U-2B\geq KT$,  the capacity of linear, full support-rank $(M,N,T,K)$-PIR  with $U$ robust  and $B$ Byzantine servers   is   $R=\frac{N-U-2B}{N-U}(1+\frac{K+T-1}{N-U-2B}+\cdots+(\frac{K+T-1}{N-U-2B})^{M-1})^{-1}$. When $K=1$ or $T=1$, the same expression gives the
capacity without assuming linearity or full support-rank.
\end{theorem}

We now prove the theorem by presenting an achievable scheme in this section. 
The converse bound will be proved in Section \ref{sec:converse}. Before presenting the general scheme, we provide an illustrative example with small parameters to demonstrate the main ideas.

\subsection{An illustrative  example for $(U,B,M, N,T,K)=(1,1,2,7,1,2)$}
Assume $W^{[1]}$ is the desired file.
Each file is comprised of $L=8$ independent symbols from $\mathbb{F}_{7}$:
\begin{align*}
    W^{[1]}=({\bf a}_1, {\bf a}_2),\quad W^{[2]}=({\bf b}_1, {\bf b}_2), 
\end{align*}
where ${\bf a}_1, {\bf a}_2,{\bf b}_1, {\bf b}_2\in \mathbb{F}_{7}^{4\times1}$. Each file is stored in the system via an MDS code with generator matrix $\mathbf{G}$. Let 
\begin{align*}
    {\bf G}=\left(\begin{array}{ccccccc}
        1 &1&1&1&1&1&1  \\
         0&1&2&3&4 &5&6
    \end{array}\right)_{2\times 7}.
\end{align*}
The storage is specified as
\begin{equation*}
 (W^{[m]}_{1},\dots, W^{[m]}_{7})=W^{[m]}{\bf G},
\end{equation*}
where $W^{[m]}_{n}\in \mathbb{F}_{7}^{4\times1}$ for $m\in[2],n\in[7]$.

The user privately chooses two random matrices $S^{[1]}$, $S^{[2]}$, independently and uniformly from the set of all full-rank $4\times 4$ matrices over $\mathbb{F}_{7}$.  
Next, we need two matrices ${\mathbf{H}}$ and ${\mathbf{H}}'$ to generate the query vectors, where ${\mathbf{H}}$ is used for the desired file and ${\mathbf{H}}'$ is used for the undesired file. 

Let ${\mathbf{H}}$ be a $2\times 7$ matrix over $\mathbb{F}_7$ satisfying the $(1,1)$-uniform selection full-rank property, and let $\mathbf{G}\star{\mathbf{H}}$ satisfy the $(1,1)$-uniform puncturing MDS property. Define $\overline{\mathbf{H}}\triangleq\operatorname{diag}(\mathbf{H},\mathbf{H})$. Therefore, $\overline{\mathbf{H}}$ is a $4\times 14$ matrix satisfying the $(1,2)$-uniform selection full-rank property.  Let $\mathbf{H}'$ be a $2\times 14$ matrix over $\mathbb{F}_7$ satisfying the $(1,2)$-uniform selection full-rank property, and let the first $7$ columns of $\mathbf{G}_{\langle 2 \rangle}\star\mathbf{H}'$ satisfy the $(1,1)$-uniform puncturing MDS property. The column properties of $\overline{\mathbf{H}}$ and $\mathbf{H}'$ are used to guarantee privacy, whereas those of $\mathbf{G}\star{\mathbf{H}}$ and $\mathbf{G}_{\langle 2\rangle}\star\mathbf{H}'$ are used to guarantee the correctness of the scheme. We will provide matrices $\mathbf{H}$ and $\mathbf{H}'$ satisfying these conditions at the end of this example.

Define query vectors: 
\begin{align*}
\left({\mathbf q}_{1,1}^{[1]},\dots, {\mathbf q}_{1,7}^{[1]},{\mathbf q}_{2,1}^{[1]},\dots, {\mathbf q}_{2,7}^{[1]}\right) &= S^{[1]} \cdot{\overline{\mathbf{H}}},\\
\left({\mathbf q}_{1,1}^{[2]},\dots, {\mathbf q}_{1,7}^{[2]},{\mathbf q}_{2,1}^{[2]},\dots, {\mathbf q}_{2,7}^{[2]}\right) &= S^{[2]}(:,1:2) \cdot{{\mathbf{H}}'}.
\end{align*}
For each $m\in[2]$ and $n\in[7]$, the two vectors $\mathbf{q}_{1,n}^{[m]}, \mathbf{q}_{2,n}^{[m]}\in \mathbb{F}_{7}^{4\times 1}$ are the query vectors generated by the user for server $n$ and file $W^{[m]}$.
The query structure is shown in Table \ref{table}.
Upon receiving the query vectors $\mathbf{q}_{1,n}^{[1]}, ~\mathbf{q}_{1,n}^{[2]}, ~\mathbf{q}_{2,n}^{[1]}+\mathbf{q}_{2,n}^{[2]}$ (this is a formal 
sum and does not represent an actual summation of the query vectors), server $n$ computes inner products $\langle\mathbf{q}_{1,n}^{[1]}, W_n^{[1]}\rangle$, $\langle\mathbf{q}_{1,n}^{[2]}, W_n^{[2]}\rangle$,  and $\langle\mathbf{q}_{2,n}^{[1]}, W_n^{[1]}\rangle+\langle\mathbf{q}_{2,n}^{[2]}, W_n^{[2]}\rangle$.

\begin{table}[t]
\centering
\renewcommand{\arraystretch}{1.3}  
\caption{Query structure.}
\label{table}
\begin{tabular}{|c| c| c| c| c| c| c|}
\hline
Server $1$ & Server $2$ & Server $3$ & Server $4$ & Server $5$ & Server $6$ &Server $7$\\
\hline
$\mathbf{q}_{1,1}^{[1]}$ & $\mathbf{q}_{1,2}^{[1]}$ & $\mathbf{q}_{1,3}^{[1]}$ & $\mathbf{q}_{1,4}^{[1]}$ & $\mathbf{q}_{1,5}^{[1]}$ & $\mathbf{q}_{1,6}^{[1]}$& $\mathbf{q}_{1,7}^{[1]}$\\
$\mathbf{q}_{1,1}^{[2]}$ & $\mathbf{q}_{1,2}^{[2]}$ & $\mathbf{q}_{1,3}^{[2]}$ & $\mathbf{q}_{1,4}^{[2]}$ & $\mathbf{q}_{1,5}^{[2]}$ & $\mathbf{q}_{1,6}^{[2]}$& $\mathbf{q}_{1,7}^{[2]}$\\ 
$\mathbf{q}_{2,1}^{[1]}+\mathbf{q}_{2,1}^{[2]}$ & $\mathbf{q}_{2,2}^{[1]}+\mathbf{q}_{2,2}^{[2]}$ & $\mathbf{q}_{2,3}^{[1]}+\mathbf{q}_{2,3}^{[2]}$ & $\mathbf{q}_{2,4}^{[1]}+\mathbf{q}_{2,4}^{[2]}$ & $\mathbf{q}_{2,5}^{[1]}+\mathbf{q}_{2,5}^{[2]}$ & $\mathbf{q}_{2,6}^{[1]}+\mathbf{q}_{2,6}^{[2]}$& $\mathbf{q}_{2,7}^{[1]}+\mathbf{q}_{2,7}^{[2]}$\\
\hline
\end{tabular}
\end{table}

We now briefly analyze how the queries guarantee privacy.
Since both  $\overline{\mathbf{H}}$ and $\mathbf{H}'$ satisfy the $(1,2)$-uniform selection full-rank property, i.e., each of them is divided into two $7$-column blocks, selecting the column at the same position from both blocks yields a matrix of full column rank. Moreover, since $S^{[1]}$ and $S^{[2]}$ are independent and uniformly distributed invertible matrices, by Lemma~\ref{lem:dim}, any server observes that the queries it receives regarding the desired file and the undesired file are $2$ uniformly random linearly independent vectors in $\mathbb{F}_7^{4}$. Consequently, the server cannot distinguish which file is being retrieved.

We now turn to the correctness of the  scheme.
For the  undesired file $W^{[2]}$, by   Claim \ref{claim}, the $14$ queried undesired symbols satisfy the following relation:
\begin{align}\label{examp1:undesired}
\left(\langle \mathbf{q}_{1,1}^{[2]}, W^{[2]}_{1} \rangle, \dots,
        \langle \mathbf{q}_{1,7}^{[2]}, W^{[2]}_{7} \rangle, 
\langle \mathbf{q}_{2,1}^{[2]}, W^{[2]}_{1} \rangle, \dots,
        \langle \mathbf{q}_{2,7}^{[2]}, W^{[2]}_{7} \rangle \right) 
        = \left( \mathbf{b}_{1}^\top S^{[2]}(:,1:2) ,  \mathbf{b}_{2}^\top S^{[2]}(:,1:2) \right) 
        \cdot (\mathbf{G}_{\langle 2\rangle} \star \mathbf{H}').
\end{align}
Specifically, the first $7$ coordinates in \eqref{examp1:undesired} correspond to the one individually downloaded interference layer
$\mathbf{q}^{[2]}_{1,n}$, while the last seven coordinates correspond to the second interference layer
$\mathbf{q}^{[2]}_{2,n}$, which is not downloaded individually but appears only in the mixed sums with the desired symbols.
Because one server is unresponsive, the user actually downloads only $6$ symbols from the set of the first $7$ symbols in \eqref{examp1:undesired}. Since the first $7$ columns of $\mathbf{G}_{\langle 2\rangle}\star \mathbf{H}'$  satisfy the $(1,1)$-uniform puncturing MDS property,  puncturing any $U=1$ column  of $(\mathbf{G}_{\langle 2 \rangle}\star \mathbf{H}')(:,1:7)$ yields a $4\times 6$ column-MDS matrix, which  can correct one error by Lemma \ref{lem:error}. The downloaded $6$ symbols from  the first seven symbols in \eqref{examp1:undesired} contain at most one error due to one Byzantine server. Thus,   $( \mathbf{b}_{1}^\top S^{[2]}(:,1:2), \mathbf{b}_{2}^\top S^{[2]}(:,1:2) )$ can be recovered, and consequently the seven symbols $\{\langle \mathbf{q}_{2,1}^{[2]}, W^{[2]}_{1} \rangle, \dots, \langle \mathbf{q}_{2,7}^{[2]}, W^{[2]}_{7} \rangle\}$ can be computed. The user can then eliminate the interferences from the undesired file in the received  pairwise summations  of the form $\langle\mathbf{q}_{2,n}^{[1]}, W_n^{[1]}\rangle+\langle\mathbf{q}_{2,n}^{[2]},W_n^{[2]}\rangle$ and finally recover all truthful desired symbols from responsive and non-Byzantine
servers and the errors remain confined to the $B=1$ Byzantine server.
For the desired file $W^{[1]}$ and each $\ell\in [2]$, the  queried desired symbols satisfy the following relation:
\begin{align*}
\left(\langle \mathbf{q}_{\ell,1}^{[1]}, W^{[1]}_{1} \rangle, \dots,
        \langle \mathbf{q}_{\ell,7}^{[1]}, W^{[1]}_{7} \rangle
\right) 
        = \left( \mathbf{a}_{1}^\top S^{[1]}(:,2\ell-1:2\ell) ,  \mathbf{a}_{2}^\top S^{[1]}(:,2\ell-1:2\ell) \right) 
        \cdot \left(\mathbf{G}\star \mathbf{H}\right).
\end{align*}
Since $\mathbf{G}\star \mathbf{H}$ satisfies the $(1,1)$-uniform puncturing MDS property, puncturing any one column of  $\mathbf{G}\star \mathbf{H}$ yields
a 4 × 6  column-MDS matrix. Thus, even in the presence of one unresponsive server and one Byzantine server, the vector $( \mathbf{a}_{1}^\top S^{[1]}(:,2\ell-1:2\ell) ,  \mathbf{a}_{2}^\top S^{[1]}(:,2\ell-1:2\ell) )$  can be recovered from received $6$ desired symbols among  $\{\langle \mathbf{q}_{\ell,1}^{[1]}, W^{[1]}_{1} \rangle, \dots, \langle \mathbf{q}_{\ell,7}^{[1]}, W^{[1]}_{7} \rangle\}$. Repeating the process for all $\ell \in [2]$ yields the concatenated 
vectors $\mathbf{a}_{1}^{\top} S^{[1]}$ and $\mathbf{a}_{2}^{\top} S^{[1]}$. Moreover, since $S^{[1]}$ is full rank, the entire file $W^{[1]}$ is retrieved once all $12$ downloaded desired symbols are collected. Therefore, the scheme is correct. The PIR rate is $R=\frac{8}{18}=\frac{4}{9}$.

For this example, we specify the matrices ${\mathbf{H}}$ and $\mathbf{H}'$ satisfying the required column properties as follows. 
\begin{align*}
    {\bf H}=&\left(\begin{array}{ccccccc}
     1 &1&1&1&1&1&1  \\
    0&1&4&2&2 &4&1
    \end{array}\right)_{2\times 7},\\
    {\bf H}'=&\left(\begin{array}{cccccccccccccc}
1&1&1&1&1&1&1&0&0&0&0&0&0&0  \\
0&1&4&2&2&4&1&1&1&1&1&1&1&1
\end{array}\right)_{2\times 14}.
\end{align*}

\subsection{General achievable scheme}
Assume $N-U-2B \geq KT$. Let $\delta$ be the smallest positive integer such that
$K$ divides $\delta(N-U-2B)$, i.e.,
$\delta=\frac{K}{\gcd(K,N-U-2B)}$.
Note that $KT\geq K+T-1$. When $N-U-2B>K+T-1$, let $\gamma$ and $\beta$ be the smallest positive integers
such that
$
(\gamma+\beta)(K+T-1)=\gamma(N-U-2B)$.
In the boundary case $N-U-2B=K+T-1$, the assumption
$N-U-2B\ge KT$ implies that $K=1$ or $T=1$, and hence $\delta=1$.
We set   $\gamma=1$ and $\beta=0$ in this case.

Assume that each file consists of $L=\delta (\gamma+\beta)^{M-1}(N-U-2B)$ symbols from a sufficiently large finite field $\mathbb{F}_q$. Consider $M$ independent files $W^{[1]}, \dots, W^{[M]}$, where each file 
is an $\frac{L}{K} \times K$ matrix whose entries are i.i.d. uniform over $\mathbb{F}_q$.  Each file is stored in the system via the same $[N,K]$ GRS storage code $\mathrm{GRS}_K(\boldsymbol{\alpha}, \mathbf{v})$ with generator matrix ${\bf G}_{K\times N}$, i.e.,
\begin{align*}
(W^{[m]}_{1},\dots, W^{[m]}_{N})=W^{[m]}{\bf G}, 
\end{align*}
where $W^{[m]}_{n}\in \mathbb{F}_q^{\frac{L}{K}\times1}$ for $m\in[M],n\in[N]$.
\begin{itemize}
 \item Without loss of generality, let $W^{[1]}$ be the desired file. The user privately chooses  $M$ random matrices, $S^{[1]},\dots,S^{[M]}$, independently and uniformly from the set of all full rank $\frac{L}{K}\times \frac{L}{K}$ matrices over $\mathbb{F}_q$.
\item Let ${{\mathbf{H}}}$ be a $\frac{\delta(N-U-2B)}{K} \times \delta N$  matrix  satisfying the $(T,\delta)$-uniform selection full-rank property   over $\mathbb{F}_q$ such that $ {{\mathbf H}} \star {{\mathbf{G}}}_{\langle \delta \rangle}$ is a  $\delta(N-U-2B)\times \delta N$  matrix  satisfying the $(U,\delta)$-uniform puncturing MDS property.
Note that the $(T,\delta)$-uniform selection full-rank property
requires the $\delta T$ selected columns of $\mathbf{H}$ to
be linearly independent. Since $\mathbf{H}$ has
$\delta(N-U-2B)/K$ rows, this requires $
\frac{\delta(N-U-2B)}{K} \ge \delta T$,
which explains the assumption $N-U-2B \ge KT$ in the proposed
construction.

\item Let $\mathbf{G}'$ be the generator matrix of a GRS code of dimension $T$ that shares the same evaluation points as the storage GRS code. By \cite[Proposition 3]{freij2017private}, $\operatorname{rowsp} (\mathbf{G}) \star \operatorname{rowsp}( \mathbf{G}' )$ is a GRS code of dimension $K+T-1$. Let $\mathbf{F} \in \mathbb{F}_q^{(K+T-1)\times N}$ be a generator matrix of $\operatorname{rowsp} (\mathbf{G}) \star \operatorname{rowsp}( \mathbf{G}' )$.
Let $\mathbf{X}$ be a $(\gamma+\beta)\times(\gamma+\beta)N$ matrix over $\mathbb{F}_q$ satisfying the $(1,\gamma+\beta)$-uniform selection full-rank property, and such that the matrix formed by the first $\gamma N$ columns of $\mathbf{X}\star\mathbf{F}_{\langle\gamma+\beta\rangle}$ satisfies the $(U,\gamma)$-uniform puncturing MDS property\footnote{The existence of ${\mathbf{H}}$ will be proved in Lemma~\ref{lem:H}, and the proof for ${\mathbf{X}}$ is similar.}. Define $\mathbf{H}' \triangleq \mathbf{X} \star \mathbf{G}'_{\langle \gamma+\beta \rangle}$, which is a $(\gamma+\beta)T \times (\gamma+\beta)N$ matrix\footnote{The properties of $\mathbf{H}'$ and $\mathbf{G}_{\langle \gamma+\beta \rangle} \star \mathbf{H}'$ will be established in Lemma~\ref{lem:H'-properties}.} over $\mathbb{F}_q$.

\item Let $\overline{\mathbf{H}}$ be an $\frac{L}{K} \times \delta (\gamma+\beta)^{M-1} N$ matrix over $\mathbb{F}_q$, 
which is block-diagonal with $(\gamma+\beta)^{M-1}$ copies of $\mathbf{H}$ on the diagonal.  Define 
\begin{align}\label{q^[1]}
\left({\mathbf q}_{1,1}^{[1]},\dots, {\mathbf q}_{1,N}^{[1]},\dots,{\mathbf q}_{\delta(\gamma+\beta)^{M-1},1}^{[1]},\dots, {\mathbf q}_{\delta(\gamma+\beta)^{M-1},N}^{[1]}\right) = S^{[1]} \cdot\overline{\mathbf{H}}.
\end{align}
Let $\overline{\mathbf{H}}'$ be a $\delta(\gamma+\beta)^{M-1}T \times \delta (\gamma+\beta)^{M-1} N$ matrix over $\mathbb{F}_q$, 
which is block-diagonal with $\delta(\gamma+\beta)^{M-2}$ copies of $\mathbf{H}'$ on the diagonal. For $m\in [2:M]$, define
\begin{align}\label{q^[i]}
\left({\mathbf q}_{1,1}^{[m]},\dots, {\mathbf q}_{1,N}^{[m]},\dots,{\mathbf q}_{\delta(\gamma+\beta)^{M-1},1}^{[m]},\dots, {\mathbf q}_{\delta(\gamma+\beta)^{M-1},N}^{[m]}\right)  = S^{[m]}\bigl(:, 1:\delta(\gamma+\beta)^{M-1}T\bigr) \cdot \overline{\mathbf{H}}',
\end{align}
where $S^{[m]}\bigl(:, 1:\delta(\gamma+\beta)^{M-1}T\bigr)$ represents the first $\delta(\gamma+\beta)^{M-1} T$ columns of $S^{[m]}$.
For each $m \in [M]$, $n \in [N]$, and $i\in [\delta(\gamma+\beta)^{M-1}]$, the vector ${\mathbf q}_{i,n}^{[m]} \in \mathbb{F}_q^{\frac{L}{K}\times 1}$ is the query computed by the user for server $n$ and file $W^{[m]}$.
\item {\it Query Structure.} For the  case $N-U-2B>K+T-1$, the query vectors are organized into
single-index queries and mixed queries as follows: For each undesired file index $m \in [2:M]$, we partition the query sequence 
\[
\left( \mathbf{q}_{1,1}^{[m]}, \dots, \mathbf{q}_{1,N}^{[m]}, \dots,
       \mathbf{q}_{\delta(\gamma+\beta)^{M-1},1}^{[m]}, \dots,
       \mathbf{q}_{\delta(\gamma+\beta)^{M-1},N}^{[m]} \right)
\]
into $\delta(\gamma+\beta)^{M-2}$ consecutive blocks, each of length $(\gamma+\beta)N$.
Each such block splits into two parts:
the first $\gamma N$ vectors are associated with an index set $\mathcal{I} \subseteq [2:M]$
that contains $m$,
and the remaining $\beta N$ vectors are associated with $\mathcal{I} \cup \{1\}$.
We call such a block an $\mathcal{I}$-block, and for a given $\mathcal{I}$ the number of
$\mathcal{I}$-blocks is $\delta \gamma^{M-|\mathcal{I}|-1} \beta^{|\mathcal{I}|-1}$.
For each such $\mathcal{I} \subseteq [2:M]$
that contains $m$, by concatenating the first parts from all $\mathcal{I}$-blocks,
we obtain a sequence of query vectors of length $\delta N \gamma^{M-|\mathcal{I}|} \beta^{|\mathcal{I}|-1}$,
which we denote by $\mathbf{q}^{[m]}_{\mathcal{I}}$.
Similarly, concatenating the second parts from all $\mathcal{I}$-blocks gives a sequence of length
$\delta N \gamma^{M-|\mathcal{I}|-1} \beta^{|\mathcal{I}|}$,
denoted by $\mathbf{q}^{[m]}_{\mathcal{I}\cup\{1\}}$.
For the desired file $W^{[1]}$, let $\mathcal{J}_1, \dots, \mathcal{J}_{2^{M-1}}$ be all subsets of $[1:M]$ that contain $1$. We partition the query sequence
\[
\left( \mathbf{q}_{1,1}^{[1]}, \dots, \mathbf{q}_{1,N}^{[1]}, \dots,
       \mathbf{q}_{\delta(\gamma+\beta)^{M-1},1}^{[1]}, \dots,
       \mathbf{q}_{\delta(\gamma+\beta)^{M-1},N}^{[1]} \right)
= S^{[1]} \cdot \overline{\mathbf{H}}=
\left(\mathbf{q}^{[1]}_{\mathcal{J}_1}, \dots, \mathbf{q}^{[1]}_{\mathcal{J}_{2^{M-1}}}\right),
\]
where each $\mathbf{q}^{[1]}_{\mathcal{J}}$ has length $\delta N \gamma^{M-|\mathcal{J}|} \beta^{|\mathcal{J}|-1}$. With a slight abuse of notation, for any nonempty $\mathcal{L}\subseteq[M]$ and any $m\in\mathcal{L}$,
let $\mathbf{q}_{\mathcal{L}}^{[m]}$ denote the corresponding  query sequence for file $W^{[m]}$
associated with $\mathcal{L}$, as defined above. 
We have the query structure: $\sum_{m \in \mathcal{L}} \mathbf{q}_{\mathcal{L}}^{[m]}$ (this is a formal 
sum and does not represent an actual summation of the query vectors). 
Arranging each $\mathbf{q}_{\mathcal{L}}^{[m]}$ as a $\delta \gamma^{M-|\mathcal{L}|} \beta^{|\mathcal{L}|-1} \times N$ matrix,
the $n$-th column of the sum matrix $\sum_{m\in\mathcal{L}} \mathbf{q}_{\mathcal{L}}^{[m]}$ gives the query sent to server $n$.
Upon receiving these query vectors, 
server $n$ computes 
\[
\sum_{m \in \mathcal{L}} 
\bigl\langle \mathbf{q}_{\mathcal{L},\,\ell,n}^{[m]},\; W_n^{[m]} \bigr\rangle ,
\]
for each $\ell \in [\delta \gamma^{M-|\mathcal{L}|} \beta^{|\mathcal{L}|-1}]$,
where $\mathbf{q}_{\mathcal{L},\,\ell,n}^{[m]}$ denotes the $\frac{L}{K}\times 1$ query vector 
in the $\ell$-th row and $n$-th column of the matrix $\mathbf{q}_{\mathcal{L}}^{[m]}$.
In the  boundary case $N-U-2B=K+T-1$, we have $\beta=0$. The mixed-query part disappears. Therefore, the query
structure reduces to the  single-layer construction.
\end{itemize}

By construction, the proposed scheme is linear. Moreover, as shown in the privacy analysis below, it also satisfies the full support-rank property.

\subsection{Correctness, privacy and rate}
We first prove the existence of ${\mathbf{H}}$. The matrix $\mathbf{H}$ in Lemma \ref{lem:H} satisfies a stronger condition, which already guarantees the required property. Hence, we only need to establish the existence of $\mathbf{H}$ as stated in the following lemma.
\begin{lemma}\label{lem:H}
  Let $\delta$ be a positive integer such that $K \mid \delta(N-U-2B)$ and $N-U-2B \ge K$.  Let $\mathbf{G} \in \mathbb{F}_q^{K \times N}$ be a column-MDS matrix.  When   $q$ is sufficiently large, there exists a matrix ${\mathbf{H}} \in \mathbb{F}_q^{\frac{\delta (N-U-2B)}{K} \times \delta N}$ such that ${\mathbf{H}}$ is column-MDS and 
$ {{\mathbf H}} \star {{\mathbf{G}}}_{\langle\delta \rangle}$ is a  $\delta(N-U-2B)\times \delta N$ column-MDS matrix.
\end{lemma}
\begin{IEEEproof}
We prove existence via the Schwartz–Zippel lemma\cite{schwartz1980fast,zippel1979probabilistic}. Treat the entries of ${\mathbf{H}} = ({h}_{ij})$ as independent variables.  We claim that for any subset $S \subseteq [\delta N]$ of size $\delta(N-U-2B)$, 
the determinant of the $\delta(N-U-2B) \times \delta(N-U-2B)$ submatrix 
$({\mathbf{H}} \star {\mathbf{G}}_{\langle\delta \rangle})_S$ formed by the columns in $S$ 
is a nonzero polynomial in the variables $h_{ij}$.
Suppose the columns of $\mathbf{G}$ are $\mathbf{g}_1, \dots, \mathbf{g}_N$.
We partition the columns of $\mathbf{G}_{\langle \delta \rangle}$ indexed by $S$ into 
$\frac{\delta(N-U-2B)}{K}$ groups, each consisting of $K$ distinct vectors from 
$\{\mathbf{g}_1,\dots,\mathbf{g}_N\}$. 
This is feasible because each original column index $n\in[N]$
is repeated at most $\delta$ times among the columns selected by $S$, and the number of groups, namely $\frac{\delta(N-U-2B)}{K}$, 
is at least $\delta$. 
By listing equal original indices consecutively
and assigning the \(\delta(N-U-2B)\) positions cyclically to
\(\frac{\delta(N-U-2B)}{K}\) groups, we can arrange the columns of $S$ into $\frac{\delta(N-U-2B)}{K}$ groups with $K$ distinct $\mathbf{g}_n$'s.  Write the resulting partition as 
$S = S_1 \cup S_2 \cup \cdots \cup S_{\delta(N-U-2B)/K}$. For each $i$, the columns of $\mathbf{G}_{\langle \delta \rangle}$ indexed by $S_i$ form a 
$K \times K$ invertible submatrix (because any $K$ columns of $\mathbf{G}$ are linearly independent and the $K$ columns indexed by $S_i$ are distinct $\mathbf{g}_n$'s).
For $i = 1,\dots,\frac{\delta(N-U-2B)}{K}$, assign values to the $i$-th row of $\mathbf{H}$ 
such that it has $1$'s in the positions $S_i$ and $0$'s elsewhere. Under this assignment, up to a column permutation, $(\mathbf{H} \star \mathbf{G}_{\langle \delta \rangle})_S$ 
is block‑diagonal with the $i$-th block being the invertible $K \times K$ matrix 
given by the columns of $\mathbf{G}_{\langle \delta\rangle }$ indexed by $S_i$. Consequently, $\det((\mathbf{H} \star \mathbf{G}_{\langle \delta \rangle})_S) \neq 0$ 
for this particular choice of $\mathbf{H}$, proving that the determinant polynomial 
in the variables $h_{ij}$ is not identically zero. This gives  ${\delta N\choose \delta (N-U-2B)}$ nonzero polynomials.

 Consider the determinant of some $\delta (N-U-2B)/K$ columns of ${{\mathbf H}}$. This gives another ${\delta N\choose \delta (N-U-2B)/K}$ nonzero polynomials.  Then based on the celebrated Schwartz-Zippel lemma, there must exist a realization of variables  over a sufficiently large finite field such that all ${\delta N\choose \delta (N-U-2B)/K}+{\delta N\choose \delta (N-U-2B)}$ polynomials simultaneously evaluate to nonzero values, and consequently a desired matrix $\mathbf{H}$ satisfying the conditions exists.
\end{IEEEproof}
\begin{remark}\label{existH'}
The existence of $\mathbf{X}$ is also proved by showing the existence of a matrix satisfying a stronger property.  Treating the entries of $\mathbf{X}$ as variables, 
the conditions that $\mathbf{X}$ satisfies the $(1,\gamma+\beta)$-uniform selection full-rank property and that the first $\gamma N$ columns 
of $\mathbf{X} \star \mathbf{F}_{\langle \gamma+\beta\rangle}$ form a column-MDS matrix 
give rise to
\[
N
+ \binom{\gamma N}{\gamma(N-U-2B)}
\]
nonzero determinant polynomials (recall that $\gamma(N-U-2B)=
(\gamma+\beta)(K+T-1)$).  When the field size $q$ is sufficiently large, there must exist a realization of variables
such that all these polynomials simultaneously evaluate to nonzero values, 
which yields the desired matrix $\mathbf{X}$. 
\end{remark}

Next, we present the properties of $\mathbf{H}'$ and $\mathbf{G}_{\langle  \gamma+\beta \rangle} \star \mathbf{H}'$.

\begin{lemma}\label{lem:H'-properties}
The matrix $\mathbf{H}'$ constructed in the above scheme  satisfies the $(T,\gamma+\beta)$-uniform selection full-rank property. Moreover,
\[
\operatorname{rowsp}\bigl(\mathbf{G}_{\langle \gamma+\beta \rangle} \star \mathbf{H}'\bigr)
=
\operatorname{rowsp}\bigl(\mathbf{X} \star \mathbf{F}_{\langle \gamma+\beta \rangle}\bigr).
\]
Since $\mathbf{X} \star \mathbf{F}_{\langle \gamma+\beta \rangle}$ has full row rank, it is a generator matrix for 
$\operatorname{rowsp}(\mathbf{G}_{\langle \gamma+\beta \rangle} \star \mathbf{H}')$.
\end{lemma}
\begin{IEEEproof}
Fix $\mathcal{T}\subseteq [N]$ with $|\mathcal{T}|=T$. We prove that the columns $
\{\mathbf{H}'(:,(\ell-1)N+n):\ n\in\mathcal{T}, \ell\in[ \gamma+\beta]\}$
are linearly independent over $\mathbb{F}_q$, where $\mathbf{H}'(:, (\ell-1)N+n)$ denotes the $((\ell-1)N+n)$-th column of $\mathbf{H}'$, which has length $(\gamma+\beta)T$. Suppose there exist
$\lambda_{\ell,n}\in\mathbb{F}_q$ such that
\begin{align}\label{independent}
\sum_{n\in \mathcal{T}}\sum_{\ell=1}^{\gamma+\beta}
\lambda_{\ell,n} \mathbf{H}'(:,(\ell-1)N+n)=\mathbf{0}.
\end{align}
We will show that $\lambda_{\ell,n}=0$ for all $n\in\mathcal{T}$ and
$ \ell \in [ \gamma+\beta]$. Recall that \(\mathbf{H}'=\mathbf{X}\star \mathbf{G}'_{\langle\gamma+\beta\rangle}\). For each \(s\in[\gamma+\beta]\) and \(t\in[T]\), the
\(((s-1)T+t)\)-th coordinate of \eqref{independent} gives
\[
\sum_{n\in\mathcal T}\sum_{\ell=1}^{\gamma+\beta}
\lambda_{\ell,n} X_{s,(\ell-1)N+n}G'_{t,n}=\sum_{n\in\mathcal T}G'_{t,n}\sum_{\ell=1}^{\gamma+\beta}
\lambda_{\ell,n} X_{s,(\ell-1)N+n}=0,
\] 
 where \(X_{s,(\ell-1)N+n}\) is the \((s,(\ell-1)N+n)\)-entry of \(\mathbf X\), and \(G'_{t,n}\) is the \((t,n)\)-entry of \(\mathbf G'\).

For each fixed \(s\), the invertibility of \(\mathbf G'(:,\mathcal T)\) implies
\[
\sum_{\ell=1}^{\gamma+\beta}
\lambda_{\ell,n} X_{s,(\ell-1)N+n}=\sum_{\ell=1}^{\gamma+\beta}
X_{s,(\ell-1)N+n} \lambda_{\ell,n} =0,\qquad n\in\mathcal T.
\] 
Because \(\mathbf{X}\) satisfies the \((1,\gamma+\beta)\)-uniform selection full-rank property, the coefficient matrix in the above equation is invertible. Therefore, \(\lambda_{\ell,n}=0\) for all \(\ell\in[\gamma+\beta]\) and \(n\in\mathcal{T}\). Hence, the \((\gamma+\beta)T\) selected columns of \(\mathbf{H}'\) are linearly independent, and \(\mathbf{H}'\) satisfies the \((T,\gamma+\beta)\)-uniform selection full-rank property.

Denote the rows of $\mathbf{G}, \mathbf{G}', \mathbf{F}, \mathbf{X}$ by $\mathbf{g}_k, \mathbf{g}'_t, \mathbf{f}_j, \mathbf{x}_u$, respectively. Every row of $\mathbf{G}_{\langle \gamma+\beta \rangle} \star \mathbf{H}'$ is of the form
$(\mathbf{g}_k)_{\langle \gamma+\beta\rangle } \star \mathbf{x}_u \star (\mathbf{g}'_t)_{\langle \gamma+\beta\rangle }=\mathbf{x}_u\star (\mathbf{g}_k\star \mathbf{g}'_t)_{\langle \gamma+\beta\rangle } $. Since
$\mathbf{G}\star \mathbf{G}'$ and $\mathbf{F}$ generate the same  code, 
\[
\operatorname{span}\{\mathbf{g}_k\star \mathbf{g}'_t \mid k\in[K],\, t\in[T]\}
=
\operatorname{span}\{\mathbf{f}_j \mid j\in[K+T-1]\}.
\]
For each fixed \(u\), multiplication component-wise by \(\mathbf{x}_u\) is a linear map. Therefore,
\[
\operatorname{span}\{{\mathbf{x}}_u\star(\mathbf{g}_k\star\mathbf{g}'_t)_{\langle\gamma+\beta\rangle}\mid k\in[K],\,t\in[T]\}
=
\operatorname{span}\{\mathbf{x}_u\star(\mathbf{f}_j)_{\langle \gamma+\beta\rangle }\mid j\in[K+T-1]\}.
\]
Taking the span over all \(u\in[\gamma+\beta]\) proves $\operatorname{rowsp}(\mathbf{G}_{\langle \gamma+\beta \rangle} \star \mathbf{H}')
=
\operatorname{rowsp}(\mathbf{X} \star \mathbf{F}_{\langle \gamma+\beta \rangle})$. Since $\mathbf{X} \star \mathbf{F}_{\langle \gamma+\beta \rangle}$ has full row rank, it is a generator matrix for 
$\operatorname{rowsp}(\mathbf{G}_{\langle \gamma+\beta \rangle} \star \mathbf{H}')$.
\end{IEEEproof}

Next, we prove the correctness and privacy of the scheme.
\subsubsection{Correctness}
For each undesired file index $m \in [2:M]$, by the block-diagonal structure of 
$\overline{\mathbf{H}}'$ in \eqref{q^[i]}, we can partition the query sequence 
into $\delta(\gamma+\beta)^{M-2}$ consecutive blocks of $(\gamma+\beta)N$ vectors.
For each $j \in [\delta(\gamma+\beta)^{M-2}]$,
\[
\left( \mathbf{q}_{(j-1)(\gamma+\beta)+1,1}^{[m]}, \dots, \mathbf{q}_{(j-1)(\gamma+\beta)+1,N}^{[m]}, \dots,
       \mathbf{q}_{j(\gamma+\beta),1}^{[m]}, \dots, \mathbf{q}_{j(\gamma+\beta),N}^{[m]} \right)
= S^{[m]}_{j} \cdot \mathbf{H}',
\]
where $S^{[m]}_{j}$ denotes the submatrix of $S^{[m]}$ consisting of columns 
$(\gamma+\beta)T(j-1)+1$ through $(\gamma+\beta)T j$.
By Claim \ref{claim}, the $(\gamma +\beta) N$ queried undesired symbols \[\left(\langle \mathbf{q}_{(j-1)(\gamma+\beta)+1,1}^{[m]}, W^{[m]}_{1} \rangle, \dots,
        \langle \mathbf{q}_{(j-1)(\gamma+\beta)+1,N}^{[m]}, W^{[m]}_{N} \rangle, \dots, 
 \langle \mathbf{q}_{j(\gamma+\beta),1}^{[m]}, W^{[m]}_{1} \rangle, \dots,
        \langle \mathbf{q}_{j(\gamma+\beta),N}^{[m]}, W^{[m]}_{N} \rangle\right)\]
  form a codeword in $\operatorname{rowsp}(\mathbf{G}_{\langle\gamma+\beta\rangle}\star\mathbf{H}')$.  By Lemma \ref{lem:H'-properties},
$\mathbf{X}\star\mathbf{F}_{\langle\gamma+\beta\rangle}$ is a generator matrix of this code. Therefore, there exists a unique $
(c_{j,1}^{[m]}, c_{j,2}^{[m]},\dots,c_{j,\gamma(N-U-2B)}^{[m]})\in \mathbb{F}_q^{\gamma(N-U-2B)}$
such that
\begin{equation}
\begin{aligned}\label{undesired}
&\bigl(\langle \mathbf{q}_{(j-1)(\gamma+\beta)+1,1}^{[m]}, W^{[m]}_{1} \rangle, \dots,
        \langle \mathbf{q}_{(j-1)(\gamma+\beta)+1,N}^{[m]}, W^{[m]}_{N} \rangle, \dots, \\
&\qquad \langle \mathbf{q}_{j(\gamma+\beta),1}^{[m]}, W^{[m]}_{1} \rangle, \dots,
        \langle \mathbf{q}_{j(\gamma+\beta),N}^{[m]}, W^{[m]}_{N} \rangle \bigr) 
        = \bigl( c_{j,1}^{[m]}, c_{j,2}^{[m]},\dots,c_{j,\gamma(N-U-2B)}^{[m]} \bigr) 
        \cdot ( \mathbf{X}\star\mathbf{F}_{\langle\gamma+\beta\rangle} ).
\end{aligned}
\end{equation}
Since the matrix $\mathbf{X}\star\mathbf{F}_{\langle\gamma+\beta\rangle}$ is a $\gamma(N-U-2B) \times (\gamma+\beta)N$ matrix and the first $\gamma N$ columns satisfy the $(U,\gamma)$-uniform puncturing MDS property, it guarantees that even in the presence of $U$ unresponsive servers and $B$ Byzantine servers, the $\gamma(N-U)$ downloaded symbols among the first $\gamma N$ symbols in \eqref{undesired} can recover 
$(  c_{j,1}^{[m]}, c_{j,2}^{[m]},\dots,c_{j,\gamma(N-U-2B)}^{[m]} )$.
This implies that the last $\beta N$ symbols in \eqref{undesired} can be recovered. The user is then able to eliminate the interferences from the undesired files in the summations and finally recovers all truthful  desired symbols from the responsive and non-Byzantine servers and   the erroneous information is still confined to the $B$ Byzantine servers.

By the block-diagonal structure of $\overline{\mathbf{H}}$ in \eqref{q^[1]},
for each $\ell \in [(\gamma+\beta)^{M-1}]$,
\[
\left( \mathbf{q}_{(\ell-1)\delta+1,1}^{[1]}, \dots, \mathbf{q}_{(\ell-1)\delta+1,N}^{[1]},\dots,\mathbf{q}_{\ell\delta ,1}^{[1]},\dots, \mathbf{q}_{\ell\delta ,N}^{[1]} \right)
= S_{\ell}^{[1]} \cdot \mathbf{H},
\]
where $S_{\ell}^{[1]}$ denotes the submatrix of $S^{[1]}$ consisting of columns 
$\frac{\delta(N-U-2B)}{K}(\ell-1)+1$ through $\frac{\delta(N-U-2B)}{K}\ell$. Let $W^{[1]}=(\mathbf{a}_{11},\dots,\mathbf{a}_{1K})$, where $\mathbf{a}_{1k} \in \mathbb{F}_q^{\frac{L}{K}\times 1}$ are the columns of $W^{[1]}$.
By Claim \ref{claim}, the $\delta N$ queried desired symbols satisfy the following relation:
\begin{equation}
\begin{aligned}\label{desired}
&\bigl(\langle \mathbf{q}_{(\ell-1)\delta+1,1}^{[1]}, W^{[1]}_{1} \rangle, \dots,
        \langle \mathbf{q}_{(\ell-1)\delta+1,N}^{[1]}, W^{[1]}_{N} \rangle, \dots, \\
&\qquad \langle \mathbf{q}_{\ell \delta ,1}^{[1]}, W^{[1]}_{1} \rangle, \dots,
        \langle \mathbf{q}_{\ell \delta ,N}^{[1]}, W^{[1]}_{N} \rangle \bigr) 
        = \bigl( \mathbf{a}_{11}^\top S^{[1]}_{\ell}, \dots, \mathbf{a}_{1K}^\top S^{[1]}_{\ell} \bigr) 
        \cdot (\mathbf{G}_{\langle\delta\rangle} \star \mathbf{H}).
\end{aligned}
\end{equation}
Because there are $U$ unresponsive servers, the user actually receives only 
$\delta(N-U)$ symbols in \eqref{desired}. 
Since ${\mathbf{G}}_{\langle \delta\rangle}\star {\mathbf{H}}$ is a $\delta(N-U-2B) \times \delta N$ 
 matrix that satisfies the $(U,\delta)$-uniform puncturing MDS property,  uniformly puncturing the same $U$ columns from every $N$-column block yields a $\delta(N-U-2B) \times \delta(N-U)$ matrix which is column-MDS. Hence, it can correct up to $\delta B$ erroneous symbols. 
Together with the fact that at most $\delta B$ of the downloaded symbols may be 
corrupted by the $B$ Byzantine servers, the vector 
$( \mathbf{a}_{11}^{\top} S^{[1]}_{\ell}, \dots, \mathbf{a}_{1K}^{\top} S^{[1]}_{\ell} )$ 
can be uniquely recovered.
Repeating the process for all $\ell \in [(\gamma+\beta)^{M-1}]$ yields the concatenated 
vector $( \mathbf{a}_{1k}^{\top} S^{[1]}_{1}, \dots, \mathbf{a}_{1k}^{\top} S^{[1]}_{(\gamma+\beta)^{M-1}} ) 
= \mathbf{a}_{1k}^{\top} S^{[1]}$ for each $k \in [K]$. 
Since $S^{[1]}$ is a full‑rank matrix, the entire file $W^{[1]}$ is retrieved once all $\delta(N-U)(\gamma+\beta)^{M-1}$ downloaded desired symbols from the set 
$\{\langle \mathbf{q}_{i,n}^{[1]}, W^{[1]}_{n} \rangle : i\in [\delta(\gamma+\beta)^{M-1}],\, n\in [N]\}$ are collected. Therefore the scheme is correct.

\begin{remark}
For simplicity of notation, \eqref{undesired} describes the queried undesired symbols
corresponding to a single undesired file $W^{[m]}$, where $m\in[2:M]$.
In the actual retrieval process, each mixed query may involve a
different subset of undesired files. Since all undesired files employ
the same matrix $\mathbf{H}'$, the corresponding interference
contributions are encoded by the same matrix $
\mathbf{X}\star\mathbf{F}_{\langle\gamma+\beta\rangle}$.
Therefore, for any mixed query, the sum of the interference
contributions from the involved undesired files preserves the same code
structure as \eqref{undesired}. More precisely, the individual coefficient vectors 
in \eqref{undesired} can be replaced by the summation of the coefficient vectors of
the involved undesired files.
\end{remark}

\subsubsection{Privacy}
To prove that the scheme is $T$ private, it suffices to show for $\forall \{n_1,\dots,n_{T}\}\subseteq[N]$, it holds that
\begin{align}\label{privacyT}
&\left(\mathbf{q}_{1,n_1}^{[1]},\dots,\mathbf{q}_{1,n_T}^{[1]},\dots,\mathbf{q}_{\delta(\gamma+\beta)^{M-1},n_1}^{[1]},\dots,\mathbf{q}_{\delta(\gamma+\beta)^{M-1},n_T}^{[1]} \right)\\
&\quad \sim \left(\mathbf{q}_{1,n_1}^{[m]},\dots,\mathbf{q}_{1,n_T}^{[m]},\dots,\mathbf{q}_{\delta(\gamma+\beta)^{M-1},n_1}^{[m]},\dots,\mathbf{q}_{\delta(\gamma+\beta)^{M-1},n_T}^{[m]} \right), \quad \forall m \in [2:M].\notag
\end{align}

For the desired file $W^{[1]}$, consider the block-diagonal matrix $\overline{\mathbf{H}} = \operatorname{diag}(\mathbf{H}, \dots, \mathbf{H})$. Since each diagonal block $\mathbf{H}$ is a $\frac{\delta(N-U-2B)}{K} \times \delta N$ matrix  satisfying the $(T,\delta)$-uniform selection full-rank property, $\overline{\mathbf{H}}$ satisfies the $(T,\delta(\gamma+\beta)^{M-1})$-uniform selection full-rank property. Thus, uniformly selecting the same $T$ columns from every $N$-column block of $\overline{\mathbf{H}}$ yields a $\frac{L}{K}\times \delta(\gamma+\beta)^{M-1}T$ matrix of full column rank. Combined with the uniform randomness of the invertible matrix $S^{[1]}$, according to Lemma \ref{lem:dim}, this guarantees that $(\mathbf{q}_{1,n_1}^{[1]},\dots,\mathbf{q}_{1,n_T}^{[1]},\dots,\mathbf{q}_{\delta(\gamma+\beta)^{M-1},n_1}^{[1]},\dots,\mathbf{q}_{\delta(\gamma+\beta)^{M-1},n_T}^{[1]} ) $ are identically distributed as $S^{[1]}(:,1:\delta(\gamma+\beta)^{M-1}T)$.

For each undesired file $W^{[m]}$, $m\in [2:M]$, consider the block-diagonal matrix $\overline{\mathbf{H}}' = \operatorname{diag}(\mathbf{H}', \dots, \mathbf{H}')$.  By Lemma \ref{lem:H'-properties}, each diagonal block $\mathbf{H}'$ is a $(\gamma+\beta) T \times (\gamma+\beta) N$ matrix  satisfying the $(T,\gamma+\beta)$-uniform selection full-rank property, therefore $\overline{\mathbf{H}}'$ satisfies the $(T,\delta(\gamma+\beta)^{M-1})$-uniform selection full-rank property. Combined with the uniform randomness of the invertible matrix $S^{[m]}$, this guarantees that $(\mathbf{q}_{1,n_1}^{[m]},\dots,\mathbf{q}_{1,n_T}^{[m]},\dots,\mathbf{q}_{\delta(\gamma+\beta)^{M-1},n_1}^{[m]},\dots,\mathbf{q}_{\delta(\gamma+\beta)^{M-1},n_T}^{[m]} )$ are identically distributed as $S^{[m]}(:,1:\delta(\gamma+\beta)^{M-1}T)$.  Furthermore, since $S^{[1]},\dots,  S^{[M]}$ are independent and identically distributed, we arrive at \eqref{privacyT}.

Since the query vectors corresponding to each file received by any set of at most $T$  servers are linearly independent, the proposed scheme satisfies the full support-rank property.

\subsubsection{Rate}
 For the  case $N-U-2B>K+T-1$,
the PIR rate is computed as follows:
\begin{align*}
R = \frac{L}{D}
&= \frac{\delta (\gamma+\beta)^{M-1}(N-U-2B)}{\delta (N-U) \sum_{d=1}^{M} \gamma^{M-d} \beta^{d-1} \binom{M}{d}} \\
&= \frac{N-U-2B}{N-U}\cdot \frac{(\gamma+\beta)^{M-1}}{\beta^{-1}\sum_{d=1}^{M} \gamma^{M-d} \beta^{d} \binom{M}{d}} \\
&= \frac{N-U-2B}{N-U}\cdot \frac{(\gamma+\beta)^{M-1}}{\beta^{-1}\left((\gamma+\beta)^{M}-\gamma^{M}\right)} \\
&= \frac{N-U-2B}{N-U}\cdot \frac{1-\frac{\gamma}{\gamma+\beta}}{1-\left(\frac{\gamma}{\gamma+\beta}\right)^{M}}\\
&=\frac{N-U-2B}{N-U}\left(1+\frac{K+T-1}{N-U-2B}+\cdots+\left(\frac{K+T-1}{N-U-2B}\right)^{M-1}\right)^{-1}.
\end{align*}

For the boundary case $N-U-2B=K+T-1$, we have $\beta=0$. Although the above result still holds for this case, the derivation involves $\beta^{-1}$. Hence, we recompute the rate for this case as follows:
\begin{align*}
R = \frac{L}{D}
&= \frac{N-U-2B}{ (N-U) M} = \frac{N-U-2B}{ N-U}\cdot \frac{1}{M}.
\end{align*}

\begin{remark}
When $K=1$ or $T=1$, the capacity results in Theorems~\ref{theorem:U-B} and~\ref{thm:multi-file-UB} hold for arbitrary $[N,K]$ MDS storage codes. The reason for choosing $\mathbf{G}$ and $\mathbf{G}'$ as generator matrices of GRS codes with the same evaluation points is to ensure that $\operatorname{rowsp}(\mathbf{G})\star\operatorname{rowsp}(\mathbf{G}')$ is an MDS code of dimension $K+T-1$. When $K=1$ or $T=1$, however, it suffices to take $\mathbf{G}$ and $\mathbf{G}'$ to be column-MDS matrices. Therefore, the GRS storage-code assumption is unnecessary when $K=1$ or $T=1$.
\end{remark}

To further illustrate the general scheme, we provide a small-parameter example using the notation introduced above.

\subsection{An example for $(U,B,M, N,T,K)=(1,0,3,7,2,2)$}
Assume $W^{[1]}$ is the desired file. Let $\delta$ be the smallest positive integer such that 
$K$ divides $\delta (N-U-2B)$. Let $\gamma$ and $\beta$ be the smallest positive 
integers such that $(\gamma+\beta)(K+T-1) = \gamma (N-U-2B)$. In this example,  we have $\delta = \gamma=\beta=1$. 
Each file is comprised of $L=\delta (\gamma+\beta)^{M-1}(N-U-2B)=24$ independent symbols from $\mathbb{F}_{7}$:
\begin{align*}
    W^{[1]}=({\bf a}_1, {\bf a}_2),\quad W^{[2]}=({\bf b}_1, {\bf b}_2), \quad W^{[3]}=({\bf c}_1, {\bf c}_2),
\end{align*}
where ${\bf a}_1, {\bf a}_2,{\bf b}_1, {\bf b}_2, {\bf c}_1, {\bf c}_2\in \mathbb{F}_{7}^{12\times1}$. Each file is stored in the system via a GRS code with generator matrix $\mathbf{G}$. Let 
\begin{align*}
    {\bf G}=\left(\begin{array}{ccccccc}
        1 &1&1&1&1&1&1  \\
         0&1&2&3&4 &5&6
    \end{array}\right)_{2\times 7}.
\end{align*}
The storage is specified as
\begin{equation*}
 (W^{[m]}_{1},\dots, W^{[m]}_{7})=W^{[m]}{\bf G},
\end{equation*}
where $W^{[m]}_{n}\in \mathbb{F}_{7}^{12\times1}$ for $m\in[3],n\in[7]$.

The user privately chooses three random matrices $S^{[1]}$, $S^{[2]}$, and $S^{[3]}$, independently and uniformly from the set of all full-rank $12\times 12$ matrices over $\mathbb{F}_{7}$.   Let
\begin{align*}
    {\bf H}=\left(\begin{array}{ccccccc}
     1 &1&1&1&1&1&1  \\
    0&1&2&4&3 &6&5\\
    0&1&4&2&2&1&4
    \end{array}\right)_{3\times 7}.
\end{align*}
 $\mathbf{H}$ is a $3\times 7$ matrix satisfying  the $(2,1)$-uniform selection full-rank property, and $\mathbf{G}\star \mathbf{H}$ is a $6\times 7$ matrix satisfying the $(1,1)$-uniform puncturing MDS property. 
 Let
\begin{align*}
    {\bf G}'&=\left(\begin{array}{ccccccc}
        1 &1&1&1&1&1&1  \\
         0&1&2&3&4 &5&6
    \end{array}\right)_{2\times 7},\\
    \mathbf{F}&=\left(\begin{array}{ccccccc}
        1 &1&1&1&1&1&1  \\
         0&1&2&3&4 &5&6\\
         0&1&4&2&2&4&1
    \end{array}\right)_{3\times 7},\\
    \mathbf{X}&=\left(\begin{array}{cccccccccccccc}
        1 &1&1&1&1&1&1&0&0&0&0&0&0&0  \\
         0&1&1&6&1&6&6&1&1&1&1&1&1&1
    \end{array}\right)_{2\times 14}.
\end{align*}
$\mathbf{G}'$ is a generator matrix of a GRS code of dimension $T=2$ that shares the same evaluation points as the storage GRS code. Therefore, $\operatorname{rowsp} (\mathbf{G}) \star \operatorname{rowsp}( \mathbf{G}' )$ is a GRS code of dimension $3$ and  $\mathbf{F}$ is a generator matrix of this GRS code. $\mathbf{X}$ satisfies the $(1,2)$-uniform selection full-rank property, and such that the matrix formed by the first $7$ columns of $\mathbf{X}\star\mathbf{F}_{\langle2\rangle}$ satisfies the $(1,1)$-uniform puncturing MDS property. Define $\mathbf{H}' \triangleq \mathbf{X} \star \mathbf{G}'_{\langle 2 \rangle}$, which is a $4\times 14$ matrix over $\mathbb{F}_7$.
According to Lemma \ref{lem:H'-properties}, $\mathbf{H}'$   satisfies the $(2,2)$-uniform selection full-rank property, and $
\operatorname{rowsp}(\mathbf{G}_{\langle 2 \rangle} \star \mathbf{H}')
=
\operatorname{rowsp}(\mathbf{X} \star \mathbf{F}_{\langle 2 \rangle})$.
Therefore, $\mathbf{X} \star \mathbf{F}_{\langle 2\rangle}$  is a generator matrix for 
$\operatorname{rowsp}(\mathbf{G}_{\langle 2\rangle} \star \mathbf{H}')$.

Define $\overline{\mathbf{H}} = \operatorname{diag}(\mathbf{H},\mathbf{H},\mathbf{H},\mathbf{H})$ and $\overline{\mathbf{H}}' = \operatorname{diag}(\mathbf{H}',\mathbf{H}')$, which are block-diagonal matrices with $\mathbf{H}$ appearing four times and $\mathbf{H}'$ twice on the diagonal, respectively.
Define 
\begin{align*}
\left({\mathbf q}_{1,1}^{[1]},\dots, {\mathbf q}_{1,7}^{[1]},\dots,{\mathbf q}_{4,1}^{[1]},\dots, {\mathbf q}_{4,7}^{[1]}\right) &= S^{[1]} \cdot{\overline{\mathbf{H}}},\\
\left({\mathbf q}_{1,1}^{[2]},\dots, {\mathbf q}_{1,7}^{[2]},\dots,{\mathbf q}_{4,1}^{[2]},\dots, {\mathbf q}_{4,7}^{[2]}\right) &= S^{[2]}(:,1:8) \cdot{\overline{\mathbf{H}}'},\\
\left({\mathbf q}_{1,1}^{[3]},\dots, {\mathbf q}_{1,7}^{[3]},\dots,{\mathbf q}_{4,1}^{[3]},\dots, {\mathbf q}_{4,7}^{[3]}\right) &= S^{[3]}(:,1:8) \cdot{\overline{\mathbf{H}}'}.
\end{align*}
For each $m\in[3]$ and $n\in[7]$, the four vectors $\mathbf{q}_{1,n}^{[m]}, \mathbf{q}_{2,n}^{[m]}, \mathbf{q}_{3,n}^{[m]}, \mathbf{q}_{4,n}^{[m]} \in \mathbb{F}_{7}^{12\times 1}$ are the query vectors generated by the user for server $n$ and file $W^{[m]}$.

The query structure is shown in Table~\ref{table2}. Upon receiving the query vectors of the forms $\mathbf{q}_{*,n}^{[m]}$, $\mathbf{q}_{*,n}^{[m]}+\mathbf{q}_{*,n}^{[m']}$ (with $m\neq m'$), and $\mathbf{q}_{*,n}^{[1]}+\mathbf{q}_{*,n}^{[2]}+\mathbf{q}_{*,n}^{[3]}$, server $n$ computes the corresponding inner products: $\langle\mathbf{q}_{*,n}^{[m]}, W_n^{[m]}\rangle$, $\langle\mathbf{q}_{*,n}^{[m]}, W_n^{[m]}\rangle+\langle\mathbf{q}_{*,n}^{[m']}, W_n^{[m']}\rangle$, and $\langle\mathbf{q}_{*,n}^{[1]}, W_n^{[1]}\rangle+\langle\mathbf{q}_{*,n}^{[2]}, W_n^{[2]}\rangle+\langle\mathbf{q}_{*,n}^{[3]}, W_n^{[3]}\rangle$. The construction of the scheme is now complete.

\begin{table}[t]
\centering
\fontsize{7pt}{9pt}\selectfont   
\renewcommand{\arraystretch}{1.2} 
\setlength{\tabcolsep}{3pt}        
\caption{Query structure.}
\label{table2}
\begin{tabular}{|c| c| c| c| c| c| c|}
\hline
Server $1$ & Server $2$ & Server $3$ & Server $4$ & Server $5$ & Server $6$ &Server $7$\\
\hline
$\mathbf{q}_{1,1}^{[1]}$ & $\mathbf{q}_{1,2}^{[1]}$ & $\mathbf{q}_{1,3}^{[1]}$ & $\mathbf{q}_{1,4}^{[1]}$ & $\mathbf{q}_{1,5}^{[1]}$ & $\mathbf{q}_{1,6}^{[1]}$& $\mathbf{q}_{1,7}^{[1]}$\\
$\mathbf{q}_{1,1}^{[2]}$ & $\mathbf{q}_{1,2}^{[2]}$ & $\mathbf{q}_{1,3}^{[2]}$ & $\mathbf{q}_{1,4}^{[2]}$ & $\mathbf{q}_{1,5}^{[2]}$ & $\mathbf{q}_{1,6}^{[2]}$& $\mathbf{q}_{1,7}^{[2]}$\\ 
$\mathbf{q}_{1,1}^{[3]}$ & $\mathbf{q}_{1,2}^{[3]}$ & $\mathbf{q}_{1,3}^{[3]}$ & $\mathbf{q}_{1,4}^{[3]}$ & $\mathbf{q}_{1,5}^{[3]}$ & $\mathbf{q}_{1,6}^{[3]}$& $\mathbf{q}_{1,7}^{[3]}$\\
$\mathbf{q}_{2,1}^{[1]}+\mathbf{q}_{2,1}^{[2]}$ & $\mathbf{q}_{2,2}^{[1]}+\mathbf{q}_{2,2}^{[2]}$ & $\mathbf{q}_{2,3}^{[1]}+\mathbf{q}_{2,3}^{[2]}$ & $\mathbf{q}_{2,4}^{[1]}+\mathbf{q}_{2,4}^{[2]}$ & $\mathbf{q}_{2,5}^{[1]}+\mathbf{q}_{2,5}^{[2]}$ & $\mathbf{q}_{2,6}^{[1]}+\mathbf{q}_{2,6}^{[2]}$& $\mathbf{q}_{2,7}^{[1]}+\mathbf{q}_{2,7}^{[2]}$\\
$\mathbf{q}_{3,1}^{[1]}+\mathbf{q}_{2,1}^{[3]}$ & $\mathbf{q}_{3,2}^{[1]}+\mathbf{q}_{2,2}^{[3]}$ & $\mathbf{q}_{3,3}^{[1]}+\mathbf{q}_{2,3}^{[3]}$ & $\mathbf{q}_{3,4}^{[1]}+\mathbf{q}_{2,4}^{[3]}$ & $\mathbf{q}_{3,5}^{[1]}+\mathbf{q}_{2,5}^{[3]}$ & $\mathbf{q}_{3,6}^{[1]}+\mathbf{q}_{2,6}^{[3]}$& $\mathbf{q}_{3,7}^{[1]}+\mathbf{q}_{2,7}^{[3]}$\\
$\mathbf{q}_{3,1}^{[2]}+\mathbf{q}_{3,1}^{[3]}$ & $\mathbf{q}_{3,2}^{[2]}+\mathbf{q}_{3,2}^{[3]}$ & $\mathbf{q}_{3,3}^{[2]}+\mathbf{q}_{3,3}^{[3]}$ & $\mathbf{q}_{3,4}^{[2]}+\mathbf{q}_{3,4}^{[3]}$ & $\mathbf{q}_{3,5}^{[2]}+\mathbf{q}_{3,5}^{[3]}$ & $\mathbf{q}_{3,6}^{[2]}+\mathbf{q}_{3,6}^{[3]}$& $\mathbf{q}_{3,7}^{[2]}+\mathbf{q}_{3,7}^{[3]}$\\
$\mathbf{q}_{4,1}^{[1]}+\mathbf{q}_{4,1}^{[2]}+\mathbf{q}_{4,1}^{[3]}$ & $\mathbf{q}_{4,2}^{[1]}+\mathbf{q}_{4,2}^{[2]}+\mathbf{q}_{4,2}^{[3]}$ & $\mathbf{q}_{4,3}^{[1]}+\mathbf{q}_{4,3}^{[2]}+\mathbf{q}_{4,3}^{[3]}$ & $\mathbf{q}_{4,4}^{[1]}+\mathbf{q}_{4,4}^{[2]}+\mathbf{q}_{4,4}^{[3]}$ & $\mathbf{q}_{4,5}^{[1]}+\mathbf{q}_{4,5}^{[2]}+\mathbf{q}_{4,5}^{[3]}$ & $\mathbf{q}_{4,6}^{[1]}+\mathbf{q}_{4,6}^{[2]}+\mathbf{q}_{4,6}^{[3]}$& $\mathbf{q}_{4,7}^{[1]}+\mathbf{q}_{4,7}^{[2]}+\mathbf{q}_{4,7}^{[3]}$\\
\hline
\end{tabular}
\end{table}

\subsubsection{$T=2$ privacy}
To prove that the scheme is $T=2$ private, it suffices to show for $\forall \{n_1,n_2\}\subseteq [7]$
\begin{align}
&\left(\mathbf{q}_{1,n_1}^{[1]}, \mathbf{q}_{1,n_2}^{[1]}, \mathbf{q}_{2,n_1}^{[1]},\mathbf{q}_{2,n_2}^{[1]}, \mathbf{q}_{3,n_1}^{[1]}, \mathbf{q}_{3,n_2}^{[1]},\mathbf{q}_{4,n_1}^{[1]}, \mathbf{q}_{4,n_2}^{[1]}\right)\notag\\
&\quad\quad \sim \left(\mathbf{q}_{1,n_1}^{[m]}, \mathbf{q}_{1,n_2}^{[m]}, \mathbf{q}_{2,n_1}^{[m]},\mathbf{q}_{2,n_2}^{[m]}, \mathbf{q}_{3,n_1}^{[m]}, \mathbf{q}_{3,n_2}^{[m]},\mathbf{q}_{4,n_1}^{[m]}, \mathbf{q}_{4,n_2}^{[m]}\right),\quad \forall m\in \{2,3\}.\label{ex2T=2}
\end{align}
For the desired file $W^{[1]}$, since $\overline{\mathbf{H}} = \operatorname{diag}(\mathbf{H}, \mathbf{H},\mathbf{H},\mathbf{H})$ and each diagonal block $\mathbf{H}$ satisfies the $(2,1)$-uniform selection full-rank property, $\overline{\mathbf{H}}$ is a $12\times 28$ matrix satisfying  the $(2,4)$-uniform selection full-rank property. Thus, uniformly selecting the same $T=2$ columns from every $7$-column block of $\overline{\mathbf{H}}$ yields a $12\times 8$ matrix of full column rank. The invertible matrix $S^{[1]}$ is chosen uniformly at random. According to Lemma \ref{lem:dim}, the query vectors $(\mathbf{q}_{1,n_1}^{[1]}, \mathbf{q}_{1,n_2}^{[1]}, \mathbf{q}_{2,n_1}^{[1]},\mathbf{q}_{2,n_2}^{[1]}, \mathbf{q}_{3,n_1}^{[1]}, \mathbf{q}_{3,n_2}^{[1]},\mathbf{q}_{4,n_1}^{[1]}, \mathbf{q}_{4,n_2}^{[1]})$  are identically distributed as $S^{[1]}(:,1:8)$. 
For the undesired file $W^{[m]}$, $m\in\{2,3\}$, $\overline{\mathbf{H}}' = \operatorname{diag}(\mathbf{H}', \mathbf{H}')$ and each diagonal block $\mathbf{H}'$ satisfies the $(2,2)$-uniform selection full-rank property, $\overline{\mathbf{H}}'$ is an $8\times 28$ matrix  satisfying the $(2,4)$-uniform selection full-rank property. Thus, uniformly selecting the same $T=2$ columns from every $7$-column block of $\overline{\mathbf{H}}'$ yields a $8\times 8$ matrix of full  rank. The invertible matrix $S^{[m]}$ is chosen uniformly at random. According to Lemma \ref{lem:dim}, the query vectors $(\mathbf{q}_{1,n_1}^{[m]}, \mathbf{q}_{1,n_2}^{[m]}, \mathbf{q}_{2,n_1}^{[m]},\mathbf{q}_{2,n_2}^{[m]}, \mathbf{q}_{3,n_1}^{[m]}, \mathbf{q}_{3,n_2}^{[m]},\mathbf{q}_{4,n_1}^{[m]}, \mathbf{q}_{4,n_2}^{[m]})$  are identically distributed as $S^{[m]}(:,1:8)$. Furthermore, since $S^{[1]}$ and $S^{[m]}$ are independent and identically distributed, we arrive at \eqref{ex2T=2}. Therefore, the scheme is linear and of full support-rank.

\subsubsection{Correctness}
Consider each undesired file index $m\in\{2,3\}$. For each $j \in [2]$, by Claim \ref{claim}, the $14$ queried undesired symbols $
(\langle \mathbf{q}_{2j-1,1}^{[m]}, W^{[m]}_{1} \rangle, \dots,
        \langle \mathbf{q}_{2j-1,7}^{[m]}, W^{[m]}_{7} \rangle, 
\langle \mathbf{q}_{2j,1}^{[m]}, W^{[m]}_{1} \rangle, \dots,
        \langle \mathbf{q}_{2j,7}^{[m]}, W^{[m]}_{7} \rangle )$
form a codeword in $\operatorname{rowsp}(\mathbf{G}_{\langle 2\rangle}\star\mathbf{H}')$.  By Lemma \ref{lem:H'-properties},
$\mathbf{X}\star\mathbf{F}_{\langle2\rangle}$ is a generator matrix of this code. Therefore, there exists a unique $
(c_{j,1}^{[m]}, c_{j,2}^{[m]},\dots,c_{j,6}^{[m]})\in \mathbb{F}_7^{6}$
such that
\begin{align*}
\left(\langle \mathbf{q}_{2j-1,1}^{[m]}, W^{[m]}_{1} \rangle, \dots,
        \langle \mathbf{q}_{2j-1,7}^{[m]}, W^{[m]}_{7} \rangle, 
\langle \mathbf{q}_{2j,1}^{[m]}, W^{[m]}_{1} \rangle, \dots,
        \langle \mathbf{q}_{2j,7}^{[m]}, W^{[m]}_{7} \rangle \right)  
        = (c_{j,1}^{[m]}, c_{j,2}^{[m]},\dots,c_{j,6}^{[m]}) 
        \cdot \left(\mathbf{X}\star\mathbf{F}_{\langle2\rangle}\right).
\end{align*}
For $j=1$,
\begin{align}\label{examp2:undesired}
\left(\langle \mathbf{q}_{1,1}^{[m]}, W^{[m]}_{1} \rangle, \dots,
        \langle \mathbf{q}_{1,7}^{[m]}, W^{[m]}_{7} \rangle, 
\langle \mathbf{q}_{2,1}^{[m]}, W^{[m]}_{1} \rangle, \dots,
        \langle \mathbf{q}_{2,7}^{[m]}, W^{[m]}_{7} \rangle \right)  
        = \left(c_{1,1}^{[m]}, c_{1,2}^{[m]},\dots,c_{1,6}^{[m]}\right) 
        \cdot \left(\mathbf{X}\star\mathbf{F}_{\langle2\rangle}\right).
\end{align}
Because one server is unresponsive, the user actually downloads only $6$ symbols from the set of the first seven symbols in \eqref{examp2:undesired}. Since the first seven columns of $\mathbf{X}\star\mathbf{F}_{\langle2\rangle}$ satisfy the $(1,1)$-uniform puncturing MDS property, puncturing any one column yields a $6\times 6$ invertible matrix.  Therefore, $(c_{1,1}^{[m]}, c_{1,2}^{[m]},\dots,c_{1,6}^{[m]}) $ can be recovered, and consequently the seven symbols $\{\langle \mathbf{q}_{2,1}^{[m]}, W^{[m]}_{1} \rangle, \dots, \langle \mathbf{q}_{2,7}^{[m]}, W^{[m]}_{7} \rangle\}$ can be computed. The user can then eliminate the interferences from the undesired files in the pairwise summations of the form $\langle\mathbf{q}_{*,n}^{[1]}, W_n^{[1]}\rangle+\langle\mathbf{q}_{*,n}^{[m]},W_n^{[m]}\rangle$.
For $j=2$, the remaining $14$ queried undesired symbols satisfy the following relation:
\begin{align}\label{ex:undesired+}
\left(\sum_{m\in\{2,3\}}\langle \mathbf{q}_{3,1}^{[m]}, W^{[m]}_{1} \rangle, \dots,
        \sum_{m\in\{2,3\}}\langle \mathbf{q}_{3,7}^{[m]}, W^{[m]}_{7} \rangle, 
\sum_{m\in\{2,3\}}\langle \mathbf{q}_{4,1}^{[m]}, W^{[m]}_{1} \rangle, \dots,
       \sum_{m\in\{2,3\}} \langle \mathbf{q}_{4,7}^{[m]}, W^{[m]}_{7} \rangle \right)\\
        = \left(c_{2,1}^{[2]}+c_{2,1}^{[3]}, c_{2,2}^{[2]}+c_{2,2}^{[3]},\dots,c_{2,6}^{[2]}+c_{2,6}^{[3]}\right) 
        \cdot \left(\mathbf{X}\star\mathbf{F}_{\langle2\rangle}\right).\notag
\end{align}
Following the same analysis as before, the vector $(c_{2,1}^{[2]}+c_{2,1}^{[3]}, c_{2,2}^{[2]}+c_{2,2}^{[3]},\dots,c_{2,6}^{[2]}+c_{2,6}^{[3]})$ can be recovered from  $6$ downloaded symbols  among the first $7$ symbols in \eqref{ex:undesired+}. Consequently, the last $7$ symbols in \eqref{ex:undesired+} can be computed. The user can then eliminate the interferences from the undesired files in the 3-wise summations of the form $\langle\mathbf{q}_{4,n}^{[1]}, W_n^{[1]}\rangle+\langle\mathbf{q}_{4,n}^{[2]},W_n^{[2]}\rangle+\langle\mathbf{q}_{4,n}^{[3]},W_n^{[3]}\rangle$ and finally recover all downloaded desired symbols.

For the desired file $W^{[1]}$ and each $\ell\in [4]$, the  queried desired symbols satisfy the following relation:
\begin{align*}
\left(\langle \mathbf{q}_{\ell,1}^{[1]}, W^{[1]}_{1} \rangle, \dots,
        \langle \mathbf{q}_{\ell,7}^{[1]}, W^{[1]}_{7} \rangle
\right) 
        = \left( \mathbf{a}_{1}^\top S^{[1]}(:,3\ell-2:3\ell) ,  \mathbf{a}_{2}^\top S^{[1]}(:,3\ell-2:3\ell) \right) 
        \cdot \left(\mathbf{G}\star \mathbf{H}\right).
\end{align*}
Since $\mathbf{G}\star \mathbf{H}$ satisfies the $(1,1)$-uniform puncturing MDS property, puncturing any one column of  $\mathbf{G}\star \mathbf{H}$ yields
a $6 \times 6$ invertible matrix. Thus, even in the presence of one unresponsive server, the vector $( \mathbf{a}_{1}^\top S^{[1]}(:,3\ell-2:3\ell) ,  \mathbf{a}_{2}^\top S^{[1]}(:,3\ell-2:3\ell) )$  can be recovered from received $6$ desired symbols among  $\{\langle \mathbf{q}_{\ell,1}^{[1]}, W^{[1]}_{1} \rangle, \dots, \langle \mathbf{q}_{\ell,7}^{[1]}, W^{[1]}_{7} \rangle\}$. Repeating the process for all $\ell \in [4]$ yields the concatenated 
vectors $\mathbf{a}_{1}^{\top} S^{[1]}$ and $\mathbf{a}_{2}^{\top} S^{[1]}$. Moreover, since $S^{[1]}$ is full rank, the entire file $W^{[1]}$ is retrieved once all $24$ downloaded desired symbols are collected. Therefore, the scheme is correct. The PIR rate is $R=\frac{24}{42}=\frac{4}{7}$.

\section{Converse bound}\label{sec:converse}
Throughout this section, we assume that $N-U-2B\geq KT$, as in Theorem~\ref{theorem:U-B}.  We first derive a general converse bound
for U-B-MDS-TPIR and then specialize it to linear,
full support-rank schemes. We also show that the
resulting capacity expression remains valid without
these restrictions when $K=1$ or $T=1$.
The proof does not require a server-symmetry assumption.

Let $
W=(W^{[1]},\ldots,W^{[M]})
$
denote the entire file set. 
The following zero-error decodability property follows from the argument used in \cite{Byzantinebanawan2018capacity} and \cite[Lemma 1]{holzbaur2021toward}. We restate it in our notation for completeness.

\begin{lemma}\label{lem:UB-decoding}
For any zero-error U-B-MDS-TPIR scheme, any set
$\mathcal{U}\subseteq[N]$ of $U$ unresponsive servers, and any
subset
\[
\mathcal{I}\subseteq[N]\setminus\mathcal{U},
\qquad
|\mathcal{I}|=N-U-2B,
\]
the desired file $W^{[m]}$ can be recovered from the truthful
answers of the servers in $\mathcal{I}$. 
\end{lemma}

\begin{IEEEproof}
Fix an arbitrary set $\mathcal{U}\subseteq[N]$ with
$|\mathcal{U}|=U$ and an arbitrary set
$\mathcal{I}\subseteq[N]\setminus\mathcal{U}$ with
$|\mathcal{I}|=N-U-2B$.
Divide the remaining $2B$ servers into two disjoint sets
$\mathcal{B}_1$ and $\mathcal{B}_2$ such that
\[
[N]\setminus(\mathcal{U}\cup\mathcal{I})
=
\mathcal{B}_1\cup\mathcal{B}_2,
\qquad
|\mathcal{B}_1|=|\mathcal{B}_2|=B.
\]

Condition on an arbitrary realization of the queries
$Q_{1:N}^{[m]}$ and suppress this realization in the notation.
For a realization $W$ of the file set, let
$A_{\mathcal{S}}^{[m]}(W)$ denote the truthful answers that would
be returned by the servers in $\mathcal{S}$.

Suppose, to the contrary, that there exist two realizations
$W$ and $\widetilde{W}$ of the file set such that
\[
W^{[m]}\neq\widetilde{W}^{[m]},
\]
but
\begin{align}
A_{\mathcal{I}}^{[m]}(W)
=
A_{\mathcal{I}}^{[m]}(\widetilde{W}).
\label{eq:same-I-answers}
\end{align}

Consider first the case where the true realization of the file set
is $W$, the servers in $\mathcal{U}$ are unresponsive, and the
Byzantine servers are indexed by $\mathcal{B}_2$.
Since Byzantine servers may return arbitrary answers, the servers
in $\mathcal{B}_2$ can return the answers corresponding to
$\widetilde{W}$. Hence, the user receives
\begin{align}
\left(
A_{\mathcal{I}}^{[m]}(W),
A_{\mathcal{B}_1}^{[m]}(W),
A_{\mathcal{B}_2}^{[m]}(\widetilde{W})
\right).
\label{eq:UB-scenario-1}
\end{align}

Next, consider the case where the true realization of the file set
is $\widetilde{W}$, the servers in $\mathcal{U}$ are unresponsive,
and the Byzantine servers are indexed by $\mathcal{B}_1$.
The servers in $\mathcal{B}_1$ can return the answers corresponding
to $W$. Hence, the user receives
\begin{align}
\left(
A_{\mathcal{I}}^{[m]}(\widetilde{W}),
A_{\mathcal{B}_1}^{[m]}(W),
A_{\mathcal{B}_2}^{[m]}(\widetilde{W})
\right).
\label{eq:UB-scenario-2}
\end{align}

By \eqref{eq:same-I-answers}, the complete answers received by the
user in \eqref{eq:UB-scenario-1} and \eqref{eq:UB-scenario-2}
are identical. However, since
$W^{[m]}\neq\widetilde{W}^{[m]}$, the same received answers
cannot be uniquely decoded to the desired file in both cases,
which contradicts the zero-error correctness requirement.
Therefore,
the desired file $W^{[m]}$ can be recovered from the truthful
answers of the servers in $\mathcal{I}$,
which completes the proof.
\end{IEEEproof}
\begin{remark}\label{rem:reduced-MDS-TPIR}
For any zero-error U-B-MDS-TPIR scheme, any set
$\mathcal{U}\subseteq[N]$ of $U$ unresponsive servers, and any subset
$\mathcal{I}\subseteq[N]\setminus\mathcal{U}$ with
$|\mathcal{I}|=N-U-2B$, the truthful responses of the servers in
$\mathcal{I}$ induce a valid $[N-U-2B,K]$ MDS-TPIR scheme.
Indeed, restricting the original storage code to $\mathcal{I}$
preserves the MDS property, decodability follows from
Lemma~\ref{lem:UB-decoding}, and the $T$-privacy constraint remains
satisfied. Hence,
$
\sum_{n\in\mathcal{I}} H(A_n^{[m]})
\geq
\frac{H(W^{[m]})}
{C_{\mathrm{MDS\text{-}TPIR}}(N-U-2B)}.$
\end{remark}
We next derive a general upper bound on the retrieval rate. The total
download satisfies
\begin{align}
D
&\geq \sum_{n\in[N]\setminus\mathcal{U}}H(A_n^{[m]}) \notag\\
&=
\frac{N-U}{N-U-2B}
\cdot
\frac{1}{\binom{N-U}{N-U-2B}}
\sum_{\substack{
\mathcal{I}\subseteq[N]\setminus\mathcal{U}\\
|\mathcal{I}|=N-U-2B
}}
\sum_{n\in\mathcal{I}}H(A_n^{[m]}) \notag\\
&\geq
\frac{N-U}{N-U-2B}
\cdot
\frac{1}{\binom{N-U}{N-U-2B}}
\sum_{\substack{
\mathcal{I}\subseteq[N]\setminus\mathcal{U}\\
|\mathcal{I}|=N-U-2B
}}
\frac{H(W^{[m]})}
{C_{\mathrm{MDS\text{-}TPIR}}(N-U-2B)}
\label{D-bounded}\\
&=
\frac{N-U}{N-U-2B}
\cdot
\frac{L}
{C_{\mathrm{MDS\text{-}TPIR}}(N-U-2B)},
\notag
\end{align}
where \eqref{D-bounded} follows from
Remark~\ref{rem:reduced-MDS-TPIR}. Consequently, any U-B-MDS-TPIR
scheme satisfies
\begin{align*}
R=\frac{L}{D}
\leq
\frac{N-U-2B}{N-U}
C_{\mathrm{MDS\text{-}TPIR}}(N-U-2B).
\end{align*}

For $K>1$ and $T>1$, we now consider an arbitrary linear,
full support-rank U-B-MDS-TPIR scheme,
the induced MDS-TPIR scheme in
Remark~\ref{rem:reduced-MDS-TPIR} remains linear and of full
support-rank. Indeed, restricting to $\mathcal I$ leaves the queries and
prescribed response functions of the retained servers unchanged. Let $C_{\mathrm{MDS\text{-}TPIR}}^{\mathrm{LFSR}}$ denote the capacity of linear and full support-rank MDS-TPIR.
Applying the converse bound for linear, full support-rank
MDS-TPIR schemes established in \cite{holzbaur2021toward}, we obtain that the rate of any linear and full-support-rank U-B-MDS-TPIR scheme satisfies
\begin{align}\label{converselinful}
R
&\leq
\frac{N-U-2B}{N-U}
C_{\mathrm{MDS\text{-}TPIR}}^{\mathrm{LFSR}}(N-U-2B)\notag\\
&\leq
\frac{N-U-2B}{N-U}
\left(
1+\frac{K+T-1}{N-U-2B}
+\cdots+
\left(\frac{K+T-1}{N-U-2B}\right)^{M-1}
\right)^{-1}.
\end{align}

For $K=1$, the converse bound for  TPIR 
is given in \cite{sun_22017capacity}, whereas for $T=1$, the converse bound for
 MDS-PIR  is given in \cite{banawan2018capacity}.
Neither bound requires linearity or full support-rank.
Applying these bounds to the general rate bound derived above
shows that the final upper bound in \eqref{converselinful}
also holds without these restrictions when $K=1$ or $T=1$.

The above upper bound analysis, together with the achievable scheme in Section~\ref{Sec:main}, proves Theorem~\ref{theorem:U-B}.

\section{Extension to Multi-File U-B-MDS-TPIR}
In this section, we extend the proposed U-B-MDS-TPIR scheme to the multi-file setting, where the user wishes to retrieve $P$ files simultaneously from $M$ files without revealing the indices of the desired files to any set of up to $T$ colluding servers. The storage model and the assumptions on the robust and Byzantine servers remain unchanged.

Assume $N-U-2B\geq KT$. For \(N-U-2B>K+T-1\), the single-file scheme in Section~\ref{Sec:main} uses \(m\)-wise summations for every \(m\in[1:M]\), whereas the multi-file scheme in this section uses only \(1\)-wise and \(M\)-wise summations, with the \(M\)-wise summations encoded by a fixed \(P\times M\) column-MDS matrix \(\mathbf{D}\).  The main result of this section is summarized in the following theorem.
   \begin{theorem}\label{thm:multi-file-UB}
When $N-U-2B\geq KT$, if the storage code is a GRS code, then there exists a $P$-out-of-$M$ multi-file U-B-MDS-TPIR scheme with rate
\[
R=
\frac{N-U-2B}{N-U}
\left(
1+\frac{(K+T-1)(M-P)}{P(N-U-2B)}
\right)^{-1}.
\]
Moreover, for $P \ge \frac{M}{2}$, this rate is capacity-achieving
when either $K=1$ or $T=1$.
\end{theorem}

To prove Theorem~\ref{thm:multi-file-UB}, we first give the general construction, and then present an example to illustrate the main ideas of the achievable scheme. Finally, we derive a matching converse bound for the cases $K=1$ or $T=1$ when $P\geq \frac{M}{2}$.

\subsection{Achievable scheme}
Assume $N-U-2B \geq KT$. Let $\delta$ be the smallest positive integer such that
$K$ divides $\delta(N-U-2B)$. Note that $KT\geq K+T-1$.
When $N-U-2B>K+T-1$, let $\gamma$ and $\beta$ be the smallest positive integers
such that
$
(\gamma+\beta)(K+T-1)=\gamma(N-U-2B)$.
In the boundary case $N-U-2B=K+T-1$, the assumption
$N-U-2B\ge KT$ implies that $K=1$ or $T=1$, and hence $\delta=1$.
We set  $\gamma=1$ and $\beta=0$ in this case. 

Assume that each file consists of $L=\delta (\gamma+\beta)(N-U-2B)$ symbols from a sufficiently large finite field $\mathbb{F}_q$. Consider $M$ independent files $W^{[1]}, \dots, W^{[M]}$, where each file 
is an $\frac{L}{K} \times K$ matrix whose entries are i.i.d. uniform over $\mathbb{F}_q$.  Each file is stored in the system via the same $[N,K]$ GRS storage code $\mathrm{GRS}_K(\boldsymbol{\alpha}, \mathbf{v})$ with generator matrix ${\bf G}_{K\times N}$, i.e.,
\begin{align*}
(W^{[m]}_{1},\dots, W^{[m]}_{N})=W^{[m]}{\bf G}, 
\end{align*}
where $W^{[m]}_{n}\in \mathbb{F}_q^{\frac{L}{K}\times1}$ for $m\in[M],n\in[N]$.
\begin{itemize}
 \item Without loss of generality, let $\mathcal{P}=[1:P]$ denote the set of indices of the desired files, i.e., $W^{[\mathcal{P}]}=\{W^{[1]},\ldots,W^{[P]}\}$ are the desired files. The user privately chooses  $M$ random matrices, $S^{[1]},\dots,S^{[M]}$, independently and uniformly from the set of all full rank $\frac{L}{K}\times \frac{L}{K}$ matrices over $\mathbb{F}_q$.
\item Let ${{\mathbf{H}}}$ be a $\frac{\delta(N-U-2B)}{K} \times \delta N$  matrix  satisfying the $(T,\delta)$-uniform selection full-rank property   over $\mathbb{F}_q$ such that $ {{\mathbf H}} \star {{\mathbf{G}}}_{\langle \delta \rangle}$ is a  $\delta(N-U-2B)\times \delta N$  matrix  satisfying the $(U,\delta)$-uniform puncturing MDS property.

\item Let $\mathbf{G}'$ be the generator matrix of a GRS code of dimension $T$ that shares the same evaluation points as the storage GRS code. By \cite[Proposition 3]{freij2017private}, $\operatorname{rowsp} (\mathbf{G}) \star \operatorname{rowsp}( \mathbf{G}' )$ is a GRS code of dimension $K+T-1$. Let $\mathbf{F} \in \mathbb{F}_q^{(K+T-1)\times N}$ be a generator matrix of $\operatorname{rowsp} (\mathbf{G}) \star \operatorname{rowsp}( \mathbf{G}' )$.
Let $\mathbf{X}$ be a $(\gamma+\beta)\times(\gamma+\beta)N$ matrix over $\mathbb{F}_q$ satisfying the $(1,\gamma+\beta)$-uniform selection full-rank property, and such that the matrix formed by the first $\gamma N$ columns of $\mathbf{X}\star\mathbf{F}_{\langle\gamma+\beta\rangle}$ satisfies the $(U,\gamma)$-uniform puncturing MDS property\footnote{The existence of $\mathbf{H}$ and $\mathbf{X}$ has already been proved in Lemma~\ref{lem:H} and Remark \ref{existH'}.}. Define $\mathbf{H}' \triangleq \mathbf{X} \star \mathbf{G}'_{\langle \gamma+\beta \rangle}$, which is a $(\gamma+\beta)T \times (\gamma+\beta)N$ matrix over $\mathbb{F}_q$.

\item Let $\overline{\mathbf{H}}$ be an $\frac{L}{K} \times \delta (\gamma+\beta) N$ matrix over $\mathbb{F}_q$, 
which is block-diagonal with $\gamma+\beta$ copies of $\mathbf{H}$ on the diagonal.  For $m\in [1:P]$, define 
\begin{align*}
\left({\mathbf q}_{1,1}^{[m]},\dots, {\mathbf q}_{1,N}^{[m]},\dots,{\mathbf q}_{\delta(\gamma+\beta),1}^{[m]},\dots, {\mathbf q}_{\delta(\gamma+\beta),N}^{[m]}\right) = S^{[m]} \cdot\overline{\mathbf{H}}.
\end{align*}
Let $\overline{\mathbf{H}}'$ be a $\delta(\gamma+\beta)T \times \delta (\gamma+\beta)N$ matrix over $\mathbb{F}_q$, 
which is block-diagonal with $\delta$ copies of $\mathbf{H}'$ on the diagonal. For $m\in [P+1:M]$, define
\begin{align*}
\left({\mathbf q}_{1,1}^{[m]},\dots, {\mathbf q}_{1,N}^{[m]},\dots,{\mathbf q}_{\delta(\gamma+\beta),1}^{[m]},\dots, {\mathbf q}_{\delta(\gamma+\beta),N}^{[m]}\right)  = S^{[m]}\bigl(:, 1:\delta(\gamma+\beta)T\bigr) \cdot \overline{\mathbf{H}}',
\end{align*}
where $S^{[m]}\bigl(:, 1:\delta(\gamma+\beta)T\bigr)$ represents the first $\delta(\gamma+\beta) T$ columns of $S^{[m]}$.
For each $m \in [M]$, $n \in [N]$, and $i\in [\delta(\gamma+\beta)]$, the vector ${\mathbf q}_{i,n}^{[m]} \in \mathbb{F}_q^{\frac{L}{K}\times 1}$ is the query computed by the user for server $n$ and file $W^{[m]}$.
\item {\it Query Structure.} For each undesired file index $m \in [P+1:M]$, we partition the query sequence 
\[
\left( \mathbf{q}_{1,1}^{[m]}, \dots, \mathbf{q}_{1,N}^{[m]}, \dots,
       \mathbf{q}_{\delta(\gamma+\beta),1}^{[m]}, \dots,
       \mathbf{q}_{\delta(\gamma+\beta),N}^{[m]} \right)
\]
into $\delta$ consecutive blocks, each of length $(\gamma+\beta)N$.
Each such block is split into two parts: the first $\gamma N$ vectors are associated with the index $m$, whereas the remaining $\beta N$ vectors are associated with the index set $[1:M]$. We call such a block an $m$-block. There are altogether $\delta$
$m$-blocks. By concatenating the first parts of all $m$-blocks, we obtain a sequence of query vectors of length $\delta \gamma N$, denoted by $\mathbf{q}^{[m]}_{m}$. Similarly, concatenating the second parts of all $m$-blocks yields a sequence of length $\delta \beta N$, denoted by $\mathbf{q}^{[m]}_{1:M}$.
For each desired file index $m \in [1:P]$,  we partition the query sequence
\[
\left( \mathbf{q}_{1,1}^{[m]}, \dots, \mathbf{q}_{1,N}^{[m]}, \dots,
       \mathbf{q}_{\delta(\gamma+\beta),1}^{[m]}, \dots,
       \mathbf{q}_{\delta(\gamma+\beta),N}^{[m]} \right)
= 
\left(\mathbf{q}^{[m]}_{m}, \mathbf{q}^{[m]}_{1:M}\right),
\]
where $\mathbf{q}^{[m]}_{m}$ has length $\delta \gamma N $ and $\mathbf{q}^{[m]}_{1:M}$ has length $\delta \beta N$.
The queries consist of single-index queries and $M$-wise queries.
For each $m\in [1:M]$, arranging  $\mathbf{q}_{m}^{[m]}$ as a $\delta \gamma \times N$ matrix,
the $n$-th column of the matrix $\mathbf{q}_{m}^{[m]}$ gives the query sent to server $n$ for $W^{[m]}$.
Let $\mathbf{D}$ be a fixed $P\times M$ column-MDS matrix. For each $m\in[1:M]$, we arrange $\mathbf{q}^{[m]}_{1:M}$ as a $\delta\beta\times N$ matrix, and denote its entry in the $\ell$-th row and $n$-th column by $\mathbf{q}^{[m]}_{1:M,\ell,n}$. For each $\ell \in [\delta\beta]$ and $n \in [N]$, the $M$-wise query sent to server $n$ is given by
\[
\mathbf{D} \cdot 
\left(
\mathbf{q}^{[1]}_{1:M,\ell,n},
\mathbf{q}^{[2]}_{1:M,\ell,n},
\dots,
\mathbf{q}^{[M]}_{1:M,\ell,n}
\right)^{\top}.
\]
\end{itemize}

\subsection{Correctness, privacy and rate}
\subsubsection{Correctness} Let $W^{[m]}=(\mathbf{a}_{m1},\dots,\mathbf{a}_{mK})$, where $m\in [M]$ and $\mathbf{a}_{mk} \in \mathbb{F}_q^{\frac{L}{K}\times 1}$ are the columns of $W^{[m]}$.

For each undesired file index $m \in [P+1:M]$ and each $j \in [\delta]$,
By Claim \ref{claim}, the $(\gamma +\beta) N$ queried undesired symbols \[\left(\langle \mathbf{q}_{(j-1)(\gamma+\beta)+1,1}^{[m]}, W^{[m]}_{1} \rangle, \dots,
        \langle \mathbf{q}_{(j-1)(\gamma+\beta)+1,N}^{[m]}, W^{[m]}_{N} \rangle, \dots, 
 \langle \mathbf{q}_{j(\gamma+\beta),1}^{[m]}, W^{[m]}_{1} \rangle, \dots,
        \langle \mathbf{q}_{j(\gamma+\beta),N}^{[m]}, W^{[m]}_{N} \rangle\right)\]
  form a codeword in $\operatorname{rowsp}(\mathbf{G}_{\langle\gamma+\beta\rangle}\star\mathbf{H}')$. By Lemma \ref{lem:H'-properties},
$\mathbf{X}\star\mathbf{F}_{\langle\gamma+\beta\rangle}$ is a generator matrix of this code. Therefore, there exists a unique $
(c_{j,1}^{[m]}, c_{j,2}^{[m]},\dots,c_{j,\gamma(N-U-2B)}^{[m]})\in \mathbb{F}_q^{\gamma(N-U-2B)}$
such that
\begin{equation}
\begin{aligned}\label{multiundesired}
&\bigl(\langle \mathbf{q}_{(j-1)(\gamma+\beta)+1,1}^{[m]}, W^{[m]}_{1} \rangle, \dots,
        \langle \mathbf{q}_{(j-1)(\gamma+\beta)+1,N}^{[m]}, W^{[m]}_{N} \rangle, \dots, \\
&\qquad \langle \mathbf{q}_{j(\gamma+\beta),1}^{[m]}, W^{[m]}_{1} \rangle, \dots,
        \langle \mathbf{q}_{j(\gamma+\beta),N}^{[m]}, W^{[m]}_{N} \rangle \bigr) 
        = \bigl( c_{j,1}^{[m]}, c_{j,2}^{[m]},\dots,c_{j,\gamma(N-U-2B)}^{[m]} \bigr) 
        \cdot ( \mathbf{X}\star\mathbf{F}_{\langle\gamma+\beta\rangle} ).
\end{aligned}
\end{equation}
Since the matrix $\mathbf{X}\star\mathbf{F}_{\langle\gamma+\beta\rangle}$ is a $\gamma(N-U-2B) \times (\gamma+\beta)N$ matrix and the first $\gamma N$ columns satisfy the $(U,\gamma)$-uniform puncturing MDS property, it guarantees that even in the presence of $U$ unresponsive servers and $B$ Byzantine servers,  the last $\beta N$ symbols in \eqref{multiundesired} can be recovered from the  $\gamma(N-U)$ downloaded symbols among the first $\gamma N$ symbols in \eqref{multiundesired}. The user is then able to eliminate the interferences from the undesired files in the $M$-wise summations. 
Since each server's $M$-wise summation is encoded using a $P\times M$ column-MDS matrix, after removing the $M-P$ undesired symbols from each summation, all $P$ desired symbols can be recovered. Consequently, every truthful desired symbol can be recovered from the responses of the responsive and non-Byzantine servers. Moreover, note that the erroneous information is still confined to the $B$ Byzantine servers.

For each desired file index $m\in [1:P]$
and  $\ell \in [\gamma+\beta]$, By Claim \ref{claim}, the $\delta N$ queried desired symbols satisfy the following relation:
\begin{equation*}
\begin{aligned}
&\bigl(\langle \mathbf{q}_{(\ell-1)\delta+1,1}^{[m]}, W^{[m]}_{1} \rangle, \dots,
        \langle \mathbf{q}_{(\ell-1)\delta+1,N}^{[m]}, W^{[m]}_{N} \rangle, \dots, \\
&\qquad \langle \mathbf{q}_{\ell \delta ,1}^{[m]}, W^{[m]}_{1} \rangle, \dots,
        \langle \mathbf{q}_{\ell \delta ,N}^{[m]}, W^{[m]}_{N} \rangle \bigr) 
        = \bigl( \mathbf{a}_{m1}^\top S^{[m]}_{\ell}, \dots, \mathbf{a}_{mK}^\top S^{[m]}_{\ell} \bigr) 
        \cdot (\mathbf{G}_{\langle\delta\rangle} \star \mathbf{H}).
\end{aligned}
\end{equation*}
where $S_{\ell}^{[m]}$ denotes the submatrix of $S^{[m]}$ consisting of columns 
$\frac{\delta(N-U-2B)}{K}(\ell-1)+1$ through $\frac{\delta(N-U-2B)}{K}\ell$. 
Since ${\mathbf{G}}_{\langle\delta\rangle}\star {\mathbf{H}}$ is a $\delta(N-U-2B) \times \delta N$ 
 matrix that satisfies the $(U,\delta)$-uniform puncturing MDS property,  uniformly puncturing the same $U$ columns from every $N$-column block yields a $\delta(N-U-2B) \times \delta(N-U)$ matrix which is column-MDS. Hence, it can correct up to $\delta B$ erroneous symbols.    
 The vector 
$( \mathbf{a}_{m1}^{\top} S^{[m]}_{\ell}, \dots, \mathbf{a}_{mK}^{\top} S^{[m]}_{\ell} )$ 
can be uniquely recovered.
Repeating the process for all $\ell \in [\gamma+\beta]$ yields the concatenated 
vector $( \mathbf{a}_{mk}^{\top} S^{[m]}_{1}, \dots, \mathbf{a}_{mk}^{\top} S^{[m]}_{\gamma+\beta}) 
= \mathbf{a}_{mk}^{\top} S^{[m]}$ for each $k \in [K]$. 
Since $S^{[m]}$ is a full‑rank matrix, the entire file $W^{[m]}$ is retrieved once all $\delta(N-U)(\gamma+\beta)$ downloaded desired symbols from the set 
$\{\langle \mathbf{q}_{i,n}^{[m]}, W^{[m]}_{n} \rangle : i\in [\delta(\gamma+\beta)],\, n\in [N]\}$ are collected.   Therefore the scheme is correct.

\subsubsection{Privacy}
To prove that the scheme is $T$ private, it suffices to show for $\forall \{n_1,\dots,n_{T}\}\subseteq[N]$, it holds that
\begin{align}\label{multiprivacyT}
&\left(\mathbf{q}_{1,n_1}^{[1]},\dots,\mathbf{q}_{1,n_T}^{[1]},\dots,\mathbf{q}_{\delta(\gamma+\beta),n_1}^{[1]},\dots,\mathbf{q}_{\delta(\gamma+\beta),n_T}^{[1]} \right)\\
&\quad \sim \left(\mathbf{q}_{1,n_1}^{[m]},\dots,\mathbf{q}_{1,n_T}^{[m]},\dots,\mathbf{q}_{\delta(\gamma+\beta),n_1}^{[m]},\dots,\mathbf{q}_{\delta(\gamma+\beta),n_T}^{[m]} \right), \quad \forall m \in [1:M].\notag
\end{align}

Similar to the privacy proof for single-file PIR in Section~\ref{Sec:main}, since both $\overline{\mathbf{H}}$ and $\overline{\mathbf{H}}'$ satisfy the $(T,\delta(\gamma+\beta))$-uniform selection full-rank property, and since $S^{[1]},\dots,S^{[M]}$ are independently and uniformly chosen full-rank matrices, it follows that for any $m\in[M]$,  $
(
\mathbf{q}_{1,n_1}^{[m]},\dots,\mathbf{q}_{1,n_T}^{[m]},\dots,
\mathbf{q}_{\delta(\gamma+\beta),n_1}^{[m]},\dots,
\mathbf{q}_{\delta(\gamma+\beta),n_T}^{[m]})$
are identically distributed as $S^{[m]}(:,1:\delta(\gamma+\beta)T)$. Therefore, \eqref{multiprivacyT} holds.

\subsubsection{Rate} 
The PIR rate is computed as follows:
\begin{align*}
R = \frac{PL}{D}
&= \frac{P\delta (\gamma+\beta)(N-U-2B)}{\delta \gamma (N-U)M+\delta \beta (N-U)P } \\
&= \frac{N-U-2B}{N-U}\cdot \frac{P(\gamma+\beta)}{\gamma M+\beta P} \\
&= \frac{N-U-2B}{N-U}\cdot \frac{P}{\left(\frac{\gamma}{\gamma+\beta}\right)M+\left(1-\frac{\gamma}{\gamma+\beta}\right)P} \\
&= \frac{N-U-2B}{N-U}\cdot \frac{P(N-U-2B)}{(N-U-2B)P+(K+T-1)(M-P)}\\
&=\frac{N-U-2B}{N-U}\left(
1+\frac{(K+T-1)(M-P)}{P(N-U-2B)}
\right)^{-1}.
\end{align*}

\subsection{An example for $(P,M,U,B, N,T,K)=(2,3,1,0,7,2,2)$}
  Each file is comprised of $L=12$ independent symbols from $\mathbb{F}_{7}$:
\begin{align*}
    W^{[1]}=({\bf a}_1, {\bf a}_2),\quad W^{[2]}=({\bf b}_1, {\bf b}_2), \quad W^{[3]}=({\bf c}_1, {\bf c}_2),
\end{align*}
where ${\bf a}_1, {\bf a}_2,{\bf b}_1, {\bf b}_2, {\bf c}_1, {\bf c}_2\in \mathbb{F}_{7}^{6\times1}$.  Each file is stored in the system via a GRS code with generator matrix $\mathbf{G}$. Let 
\begin{align*}
    {\bf G}=\left(\begin{array}{ccccccc}
        1 &1&1&1&1&1&1  \\
         0&1&2&3&4 &5&6
    \end{array}\right)_{2\times 7}.
\end{align*}
The storage is specified as
\begin{equation*}
 (W^{[m]}_{1},\dots, W^{[m]}_{7})=W^{[m]}{\bf G},
\end{equation*}
where $W^{[m]}_{n}\in \mathbb{F}_{7}^{6\times1}$ for $m\in[3],n\in[7]$.

Assume $W^{[1]}$ and $W^{[2]}$ are the desired files. The user privately chooses three random matrices $S^{[1]}$, $S^{[2]}$, and $S^{[3]}$, independently and uniformly from the set of all full-rank $6\times 6$ matrices over $\mathbb{F}_{7}$.   Let
\begin{align*}
    {\bf H}=\left(\begin{array}{ccccccc}
     1 &1&1&1&1&1&1  \\
    0&1&2&4&3 &6&5\\
    0&1&4&2&2&1&4
    \end{array}\right)_{3\times 7}.
\end{align*}
 $\mathbf{H}$ is a $3\times 7$ matrix satisfying  the $(2,1)$-uniform selection full-rank property, and $\mathbf{G}\star \mathbf{H}$ is a $6\times 7$ matrix satisfying the $(1,1)$-uniform puncturing MDS property. 
 Let
\begin{align*}
    {\bf G}'&=\left(\begin{array}{ccccccc}
        1 &1&1&1&1&1&1  \\
         0&1&2&3&4 &5&6
    \end{array}\right)_{2\times 7},\\
    \mathbf{F}&=\left(\begin{array}{ccccccc}
        1 &1&1&1&1&1&1  \\
         0&1&2&3&4 &5&6\\
         0&1&4&2&2&4&1
    \end{array}\right)_{3\times 7},\\
    \mathbf{X}&=\left(\begin{array}{cccccccccccccc}
        1 &1&1&1&1&1&1&0&0&0&0&0&0&0  \\
         0&1&1&6&1&6&6&1&1&1&1&1&1&1
    \end{array}\right)_{2\times 14}.
\end{align*}
$\mathbf{G}'$ is a generator matrix of a GRS code of dimension $T=2$ that shares the same evaluation points as the storage GRS code. Therefore, $\operatorname{rowsp} (\mathbf{G}) \star \operatorname{rowsp}( \mathbf{G}' )$ is a GRS code of dimension $3$ and  $\mathbf{F}$ is a generator matrix of this GRS code. $\mathbf{X}$ satisfies the $(1,2)$-uniform selection full-rank property, and such that the matrix formed by the first $7$ columns of $\mathbf{X}\star\mathbf{F}_{\langle2\rangle}$ satisfies the $(1,1)$-uniform puncturing MDS property. Define $\mathbf{H}' \triangleq \mathbf{X} \star \mathbf{G}'_{\langle 2 \rangle}$, which is a $4\times 14$ matrix over $\mathbb{F}_7$.
According to Lemma \ref{lem:H'-properties}, $\mathbf{H}'$   satisfies the $(2,2)$-uniform selection full-rank property, and $
\operatorname{rowsp}(\mathbf{G}_{\langle 2 \rangle} \star \mathbf{H}')
=
\operatorname{rowsp}(\mathbf{X} \star \mathbf{F}_{\langle 2 \rangle})$.
Therefore, $\mathbf{X} \star \mathbf{F}_{\langle 2\rangle}$  is a generator matrix for 
$\operatorname{rowsp}(\mathbf{G}_{\langle 2\rangle} \star \mathbf{H}')$.

Define $\overline{\mathbf{H}} = \operatorname{diag}(\mathbf{H},\mathbf{H})$, which is a block-diagonal matrix with $\mathbf{H}$ appearing   twice on the diagonal. 
Define 
\begin{align*}
\left({\mathbf q}_{1,1}^{[1]},\dots, {\mathbf q}_{1,7}^{[1]},{\mathbf q}_{2,1}^{[1]},\dots, {\mathbf q}_{2,7}^{[1]}\right) &= S^{[1]} \cdot{\overline{\mathbf{H}}},\\
\left({\mathbf q}_{1,1}^{[2]},\dots, {\mathbf q}_{1,7}^{[2]},{\mathbf q}_{2,1}^{[2]},\dots, {\mathbf q}_{2,7}^{[2]}\right) &= S^{[2]}\cdot{\overline{\mathbf{H}}},\\
\left({\mathbf q}_{1,1}^{[3]},\dots, {\mathbf q}_{1,7}^{[3]},{\mathbf q}_{2,1}^{[3]},\dots, {\mathbf q}_{2,7}^{[3]}\right) &= S^{[3]}(:,1:4) \cdot{{\mathbf{H}}'}.
\end{align*}
Let \[\mathbf{D} = \begin{pmatrix} 1 & 1 & 1 \\ 1 & 2 & 3 \end{pmatrix}\] be a fixed $2\times 3$ column-MDS matrix over $\mathbb{F}_7$.

\begin{table}[t]
\centering
\fontsize{7pt}{9pt}\selectfont   
\renewcommand{\arraystretch}{1.2} 
\setlength{\tabcolsep}{1.2pt}        
\caption{Query structure.}
\label{table3}
\begin{tabular}{|c| c| c| c| c| c| c|}
\hline
Server $1$ & Server $2$ & Server $3$ & Server $4$ & Server $5$ & Server $6$ &Server $7$\\
\hline
$\mathbf{q}_{1,1}^{[1]}$ & $\mathbf{q}_{1,2}^{[1]}$ & $\mathbf{q}_{1,3}^{[1]}$ & $\mathbf{q}_{1,4}^{[1]}$ & $\mathbf{q}_{1,5}^{[1]}$ & $\mathbf{q}_{1,6}^{[1]}$& $\mathbf{q}_{1,7}^{[1]}$\\
$\mathbf{q}_{1,1}^{[2]}$ & $\mathbf{q}_{1,2}^{[2]}$ & $\mathbf{q}_{1,3}^{[2]}$ & $\mathbf{q}_{1,4}^{[2]}$ & $\mathbf{q}_{1,5}^{[2]}$ & $\mathbf{q}_{1,6}^{[2]}$& $\mathbf{q}_{1,7}^{[2]}$\\ 
$\mathbf{q}_{1,1}^{[3]}$ & $\mathbf{q}_{1,2}^{[3]}$ & $\mathbf{q}_{1,3}^{[3]}$ & $\mathbf{q}_{1,4}^{[3]}$ & $\mathbf{q}_{1,5}^{[3]}$ & $\mathbf{q}_{1,6}^{[3]}$& $\mathbf{q}_{1,7}^{[3]}$\\
$\mathbf{q}_{2,1}^{[1]}+\mathbf{q}_{2,1}^{[2]}+\mathbf{q}_{2,1}^{[3]}$ & $\mathbf{q}_{2,2}^{[1]}+\mathbf{q}_{2,2}^{[2]}+\mathbf{q}_{2,2}^{[3]}$ & $\mathbf{q}_{2,3}^{[1]}+\mathbf{q}_{2,3}^{[2]}+\mathbf{q}_{2,3}^{[3]}$ & $\mathbf{q}_{2,4}^{[1]}+\mathbf{q}_{2,4}^{[2]}+\mathbf{q}_{2,4}^{[3]}$ & $\mathbf{q}_{2,5}^{[1]}+\mathbf{q}_{2,5}^{[2]}+\mathbf{q}_{2,5}^{[3]}$ & $\mathbf{q}_{2,6}^{[1]}+\mathbf{q}_{2,6}^{[2]}+\mathbf{q}_{2,6}^{[3]}$& $\mathbf{q}_{2,7}^{[1]}+\mathbf{q}_{2,7}^{[2]}+\mathbf{q}_{2,7}^{[3]}$\\
$\mathbf{q}_{2,1}^{[1]}+2\mathbf{q}_{2,1}^{[2]}+3\mathbf{q}_{2,1}^{[3]}$ & $\mathbf{q}_{2,2}^{[1]}+2\mathbf{q}_{2,2}^{[2]}+3\mathbf{q}_{2,2}^{[3]}$ & $\mathbf{q}_{2,3}^{[1]}+2\mathbf{q}_{2,3}^{[2]}+3\mathbf{q}_{2,3}^{[3]}$ & $\mathbf{q}_{2,4}^{[1]}+2\mathbf{q}_{2,4}^{[2]}+3\mathbf{q}_{2,4}^{[3]}$ & $\mathbf{q}_{2,5}^{[1]}+2\mathbf{q}_{2,5}^{[2]}+3\mathbf{q}_{2,5}^{[3]}$ & $\mathbf{q}_{2,6}^{[1]}+2\mathbf{q}_{2,6}^{[2]}+3\mathbf{q}_{2,6}^{[3]}$& $\mathbf{q}_{2,7}^{[1]}+2\mathbf{q}_{2,7}^{[2]}+3\mathbf{q}_{2,7}^{[3]}$\\
\hline
\end{tabular}
\end{table}

The query structure is shown in Table~\ref{table3}. Upon receiving the query vectors of the forms $\mathbf{q}_{1,n}^{[m]}$ for $m\in [3]$, $\mathbf{q}_{2,n}^{[1]}+\mathbf{q}_{2,n}^{[2]}+\mathbf{q}_{2,n}^{[3]}$ and $\mathbf{q}_{2,n}^{[1]}+2\mathbf{q}_{2,n}^{[2]}+3\mathbf{q}_{2,n}^{[3]}$, server $n$ computes the corresponding inner products: $\langle\mathbf{q}_{1,n}^{[m]}, W_n^{[m]}\rangle$ for  $m\in [3]$,  $\langle\mathbf{q}_{2,n}^{[1]}, W_n^{[1]}\rangle+\langle\mathbf{q}_{2,n}^{[2]}, W_n^{[2]}\rangle+\langle\mathbf{q}_{2,n}^{[3]}, W_n^{[3]}\rangle$ and $\langle\mathbf{q}_{2,n}^{[1]}, W_n^{[1]}\rangle+2\langle\mathbf{q}_{2,n}^{[2]}, W_n^{[2]}\rangle+3\langle\mathbf{q}_{2,n}^{[3]}, W_n^{[3]}\rangle$. The 3-wise summations are encoded by the column-MDS matrix $\mathbf{D}$. The construction of the scheme is now complete.

\subsubsection{$T=2$ privacy}
To prove that the scheme is $T=2$ private, it suffices to show for any $\{n_1,n_2\}\subseteq  [7]$
\begin{align}\label{multiT=2}
\left(\mathbf{q}_{1,n_1}^{[1]}, \mathbf{q}_{1,n_2}^{[1]},\mathbf{q}_{2,n_1}^{[1]},\mathbf{q}_{2,n_2}^{[1]}\right)\sim \left(\mathbf{q}_{1,n_1}^{[m]},\mathbf{q}_{1,n_2}^{[m]},\mathbf{q}_{2,n_1}^{[m]}, \mathbf{q}_{2,n_2}^{[m]}\right),\quad \forall m\in [3].
\end{align}
Since both $\overline{\mathbf{H}}$ and ${\mathbf{H}}'$ satisfy the $(2,2)$-uniform selection full-rank property, uniformly selecting the same $T=2$ columns from every $7$-column block  yields a matrix of full column rank.  $S^{[1]},S^{[2]},S^{[3]}$ are independently and uniformly chosen full-rank matrices. According to Lemma \ref{lem:dim}, it follows that for any $m\in[3]$,  $
(\mathbf{q}_{1,n_1}^{[m]},\mathbf{q}_{1,n_2}^{[m]},\mathbf{q}_{2,n_1}^{[m]}, \mathbf{q}_{2,n_2}^{[m]})$
are identically distributed as $S^{[m]}(:,1:4)$. Therefore, \eqref{multiT=2} holds.

\subsubsection{Correctness}
By Claim \ref{claim} and Lemma \ref{lem:H'-properties}, there exists a unique $
(c_{1}, c_{2},\dots,c_{6})\in \mathbb{F}_7^{6}$
such that
\begin{align*}
\left(\langle \mathbf{q}_{1,1}^{[3]}, W^{[3]}_{1} \rangle, \dots,
        \langle \mathbf{q}_{1,7}^{[3]}, W^{[3]}_{7} \rangle, 
\langle \mathbf{q}_{2,1}^{[3]}, W^{[3]}_{1} \rangle, \dots,
        \langle \mathbf{q}_{2,7}^{[3]}, W^{[3]}_{7} \rangle \right) 
        = \left( c_{1}, c_{2},\dots,c_{6} \right) 
        \cdot\left(\mathbf{X}\star\mathbf{F}_{\langle2\rangle}\right).
\end{align*}
 Since the first seven columns of $\mathbf{X}\star\mathbf{F}_{\langle2\rangle}$ satisfy the $(1,1)$-uniform puncturing MDS property, puncturing any one column of $(\mathbf{X}\star\mathbf{F}_{\langle2\rangle})(:,1:7)$ yields a $6\times 6$ invertible matrix.  Therefore, $( c_{1}, c_{2},\dots,c_{6})$ can be recovered from $6$ downloaded symbols among $\{\langle \mathbf{q}_{1,1}^{[3]}, W^{[3]}_{1} \rangle, \dots,
        \langle \mathbf{q}_{1,7}^{[3]}, W^{[3]}_{7} \rangle\}$, and consequently the seven symbols $\{\langle \mathbf{q}_{2,1}^{[3]}, W^{[3]}_{1} \rangle, \dots, \langle \mathbf{q}_{2,7}^{[3]}, W^{[3]}_{7} \rangle\}$ can be computed. The user can then eliminate the interferences from the undesired file in the $3$-wise summations. Since each server's $3$-wise summation is encoded using a $2\times 3$ column-MDS matrix, after removing the undesired symbol from each such summation, its two desired symbols can be recovered.

For each  desired file index $m\in [1:2]$ and each $\ell\in [1:2]$. Assume $W^{[m]}=(\mathbf{x}_1, \mathbf{x}_2)$. The  queried desired symbols satisfy the following relation:
\begin{align*}
\left(\langle \mathbf{q}_{\ell,1}^{[m]}, W^{[m]}_{1} \rangle, \dots,
        \langle \mathbf{q}_{\ell,7}^{[m]}, W^{[m]}_{7} \rangle
\right) 
        = \left( \mathbf{x}_{1}^\top S^{[m]}(:,3\ell-2:3\ell) ,  \mathbf{x}_{2}^\top S^{[m]}(:,3\ell-2:3\ell) \right) 
        \cdot \left(\mathbf{G}\star \mathbf{H}\right).
\end{align*}
Since $\mathbf{G}\star \mathbf{H}$ satisfies the $(1,1)$-uniform puncturing MDS property,  even in the presence of one unresponsive server, the vector $( \mathbf{x}_{1}^\top S^{[m]}(:,3\ell-2:3\ell) ,  \mathbf{x}_{2}^\top S^{[m]}(:,3\ell-2:3\ell) )$  can be recovered from received $6$ desired symbols among  $\{\langle \mathbf{q}_{\ell,1}^{[m]}, W^{[m]}_{1} \rangle, \dots, \langle \mathbf{q}_{\ell,7}^{[m]}, W^{[m]}_{7} \rangle\}$. Repeating the process for all $\ell \in [2]$ yields the concatenated 
vectors $\mathbf{x}_{1}^{\top} S^{[m]}$ and $\mathbf{x}_{2}^{\top} S^{[m]}$. Moreover, since $S^{[m]}$ is full rank, the entire file $W^{[m]}$ is retrieved once all $12$ downloaded desired symbols are recovered.  Therefore the scheme is correct. The PIR rate is $R=\frac{24}{30}=\frac{4}{5}$.

\subsection{The converse bound}

In this subsection, we consider $P\geq  \frac{M}{2}$ and
$N-U-2B\geq KT$. We derive a converse bound for the multi-file
U-B-MDS-TPIR problem and show that it matches the achievable rate
in Theorem~\ref{thm:multi-file-UB} when either $K=1$ or $T=1$.

Lemma~\ref{lem:UB-decoding} extends directly to the multi-file setting
by replacing the desired file $W^{[m]}$ with the  desired files
$W^{[\mathcal{P}]}=\{W^{[m]}: m\in \mathcal{P}\}$. Hence, the desired files can be recovered
from the truthful answers of any $N-U-2B$ responsive servers. Similarly, Remark~\ref{rem:reduced-MDS-TPIR} extends directly to the
multi-file setting. For any set
$\mathcal{U}\subseteq[N]$ of $U$ unresponsive servers, and any subset
$\mathcal{I}\subseteq[N]\setminus\mathcal{U}$ with
$|\mathcal{I}|=N-U-2B$, the truthful responses of the servers in
$\mathcal{I}$ induce a valid $[N-U-2B,K]$ multi-file MDS-TPIR scheme.
Therefore,
\[
\sum_{n\in\mathcal{I}} H(A_n^{[\mathcal{P}]})
\geq
\frac{H(W^{[\mathcal{P}]})}
{C_{\mathrm{multi\text{-}file\ MDS\text{-}TPIR}}(N-U-2B)}.
\]

Since
\[
H(W^{[\mathcal{P}]})=PL,
\]
the total download can be lower bounded as
\begin{align*}
D
&\geq
\sum_{n\in[N]\setminus\mathcal{U}}
H(A_n^{[\mathcal{P}]}) \\
&=
\frac{N-U}{N-U-2B}
\cdot
\frac{1}{\binom{N-U}{N-U-2B}}
\sum_{\substack{
\mathcal{I}\subseteq[N]\setminus\mathcal{U}\\
|\mathcal{I}|=N-U-2B
}}
\sum_{n\in\mathcal{I}}
H(A_n^{[\mathcal{P}]}) \\
&\geq
\frac{N-U}{N-U-2B}
\cdot
\frac{PL}
{C_{\mathrm{multi\text{-}file\ MDS\text{-}TPIR}}(N-U-2B)}.
\end{align*}

Therefore, the rate of any multi-file U-B-MDS-TPIR scheme satisfies
\begin{align}
R
&=
\frac{PL}{D} \notag\\
&\leq
\frac{N-U-2B}{N-U}
C_{\mathrm{multi\text{-}file\ MDS\text{-}TPIR}}(N-U-2B).
\label{eq:multi-file-converse}
\end{align}

For $P\geq \frac{M}{2}$, the capacity of multi-file MDS-TPIR with $K=1$ or $T=1$ was established in~\cite{zhang2017private}. For our purpose, only the upper bounds are needed. For completeness, we re-prove these upper bounds in the Appendix without imposing the symmetry assumption.
The upper bounds for multi-file TPIR and
multi-file MDS-PIR in the cases $K=1$ and $T=1$, respectively, are given by
\begin{align}\label{multiTPIR}
C_{\mathrm{multi\text{-}file\ TPIR}}(N-U-2B)
\leq
\left(
1+
\frac{T(M-P)}
{P(N-U-2B)}
\right)^{-1},
\end{align}
when $K=1$, and
\begin{align}\label{multiMDSPIR}
C_{\mathrm{multi\text{-}file\ MDS\text{-}PIR}}(N-U-2B)
\leq 
\left(
1+
\frac{K(M-P)}
{P(N-U-2B)}
\right)^{-1},
\end{align}
when $T=1$.

Applying these upper bounds  to \eqref{eq:multi-file-converse},
when $K=1$, the rate of  multi-file U-B-TPIR scheme satisfies
\[
R
\leq
\frac{N-U-2B}{N-U}
\left(
1+
\frac{T(M-P)}
{P(N-U-2B)}
\right)^{-1},
\]
and when $T=1$, the rate of multi-file U-B-MDS-PIR scheme satisfies
\[
R
\leq
\frac{N-U-2B}{N-U}
\left(
1+
\frac{K(M-P)}
{P(N-U-2B)}
\right)^{-1}.
\]
Both upper bounds coincide with the corresponding achievable rates in
Theorem~\ref{thm:multi-file-UB}. Hence, the proposed multi-file
U-B-MDS-TPIR scheme is capacity-achieving when either $K=1$ or $T=1$.

\section{Conclusion}
In this paper, we studied private information retrieval
from GRS-coded storage systems with colluding, robust,
and Byzantine servers. For $N-U-2B\ge KT$, we characterize the capacity of  linear, full
support-rank U-B-MDS-TPIR scheme based on Schur product
codes and
carefully designed query matrices, which enables private retrieval
while tolerating server failures and Byzantine errors.  When $K=1$ or $T=1$, the same capacity
expression holds without the linearity and full
support-rank restrictions.

We further extended the construction to $P$-out-of-$M$
multi-file U-B-MDS-TPIR under the same condition
$N-U-2B \ge KT$.
For $P \ge M/2$, the resulting scheme achieves capacity
when either $K=1$ or $T=1$.

A natural direction for future work is to characterize
the capacity of U-B-MDS-TPIR for $K>1$ and $T>1$ without
the linearity and full support-rank restrictions.

\section*{AI Statement}

During the preparation of this work, the authors used DeepSeek and ChatGPT for language polishing. In addition, ChatGPT was used as an auxiliary tool to assist in constructing the query matrix $\mathbf{H}'$ in the proposed protocol. The authors reviewed and verified all AI-assisted content and take full responsibility for the content of this publication.

\bibliographystyle{IEEEtran}
\bibliography{reference}

\appendix
In this appendix, we provide direct proofs of the converse bounds for
multi-file TPIR and multi-file MDS-PIR, which are used to establish the
capacity of multi-file U-B-MDS-TPIR when $K=1$ or $T=1$. Throughout this
appendix, we consider the parameter regime
\[
 P\geq \frac{M}{2},
\qquad
N'\geq KT,
\]
where $M$ denotes the total number of files, $P$ denotes the number of
desired files, and $N'$ denotes the number of servers in the
considered multi-file MDS-TPIR system.

Define 
\begin{align*}
\mathcal{Q}&\triangleq\{Q_n^{[\mathcal{P}]}: \mathcal{P}\subseteq [1:M], |\mathcal{P}|=P, n\in [1:N'] \},\\
A^{[\mathcal{P}]}_{\mathcal{X}}&\triangleq \{A_n^{[\mathcal{P}]}: n\in \mathcal{X}\}.
\end{align*}
\begin{theorem}[Han's Inequality, \cite{Cover2006}]\label{th:Han's}
For any integer $X\in[1:N']$,
\[
\frac{1}{\binom{N'}{X}}
\sum_{\substack{\mathcal{X}\subseteq[1:N']\\|\mathcal{X}|=X}}
H(A_{\mathcal{X}}^{[M-P+1:M]}|W_{1:P},\mathcal{Q})
\geq
\frac{X}{N'}
H(A_{1:N'}^{[M-P+1:M]}|W_{1:P},\mathcal{Q}).
\]
\end{theorem}
The following equality follows from \cite[Lemma 2]{banawan2018multi}.
\begin{align}
    H(A_{1:N'}^{[M-P+1:M]}|W_{1:P},\mathcal{Q})&=I(W_{P+1:M};A_{1:N'}^{[M-P+1:M]}|W_{1:P},\mathcal{Q})\notag\\
    &=H(W_{P+1:M}|W_{1:P},\mathcal{Q})-H(W_{P+1:M}|A_{1:N'}^{[M-P+1:M]},W_{1:P},\mathcal{Q})\notag\\
    &=H(W_{P+1:M}|W_{1:P},\mathcal{Q})\label{M-P+1}\\
    &=(M-P)L\label{(M-P)L},
\end{align}
where \eqref{M-P+1} follows from $P\geq \frac{M}{2}$
 and $[P+1:M]\subseteq [M-P+1:M]$.
 
 We next consider the two special cases $K=1$ and $T=1$ separately.
\subsection{Converse bound for multi-file TPIR}
Let $\mathcal{T}\subseteq [1:N']$ with $|\mathcal{T}|=T$. The total download satisfies 
\begin{align}
    D& \geq H(A_{1:N'}^{[1:P]}|\mathcal{Q})\notag\\   &=H(W_{1:P},A_{1:N'}^{[1:P]}|\mathcal{Q})\notag\\
    &\geq H(W_{1:P}|\mathcal{Q})+H(A_{\mathcal{T}}^{[1:P]}|W_{1:P},\mathcal{Q})\notag\\
    &=PL+H(A_{\mathcal{T}}^{[M-P+1:M]}|W_{1:P},\mathcal{Q})\label{A_{1:T}},
\end{align}
where \eqref{A_{1:T}} is a consequence of $T$-privacy: The joint entropy of the downloads from any set of at most $T$ servers must be identical regardless of which files are retrieved (see \cite[Lemma 2]{sun2017private} for proof details, which are given for single-file retrieval but extend directly to the multi-file setting). 

Since \eqref{A_{1:T}} holds for every $T$-element subset of $[1:N']$, summing the $\binom{N'}{T}$ inequalities and averaging over them yields
\begin{align}
D&\geq PL+\frac{1}{\binom{N'}{T}}
\sum_{\substack{\mathcal{T}\subseteq[1:N']\\|\mathcal{T}|=T}}
H(A_{\mathcal{T}}^{[M-P+1:M]}|W_{1:P}, \mathcal{Q})\notag\\
&\geq PL+\frac{T}{N'}
H(A_{1:N'}^{[M-P+1:M]}|W_{1:P},\mathcal{Q})\label{THan's}\\
&=PL+\frac{T}{N'}(M-P)L\label{D-T},
\end{align}
where \eqref{THan's} follows from Han's Inequality and \eqref{D-T} follows from \eqref{(M-P)L}. From \eqref{D-T}, we have 
\begin{align*}
C_{\mathrm{multi\text{-}file\ TPIR}}(N')
\leq
\left(
1+
\frac{T(M-P)}
{PN'}
\right)^{-1},
\end{align*}
which proves \eqref{multiTPIR}.
\subsection{Converse bound for multi-file MDS-PIR}
 Let $\mathcal{K}\subseteq [1:N']$ with $|\mathcal{K}|=K$. The total download satisfies 
\begin{align}
    D& \geq H(A_{1:N'}^{[1:P]}|\mathcal{Q})\notag\\   &=H(W_{1:P},A_{1:N'}^{[1:P]}|\mathcal{Q})\notag\\
    &\geq H(W_{1:P}|\mathcal{Q})+H(A_{\mathcal{K}}^{[1:P]}|W_{1:P},\mathcal{Q})\notag\\
    &=PL+\sum_{n\in \mathcal{K}}H(A_{n}^{[1:P]}|W_{1:P},\mathcal{Q})\label{A_1{1:K}}\\
    &=PL+\sum_{n\in \mathcal{K}}H(A_{n}^{[M-P+1:M]}|W_{1:P},\mathcal{Q})\label{{K}privacy}\\
    &=PL+H(A_{\mathcal{K}}^{[M-P+1:M]}|W_{1:P},\mathcal{Q})\label{A_2{1:K}},
\end{align}
where \eqref{{K}privacy} is a consequence of $T=1$ privacy, \eqref{A_1{1:K}} and \eqref{A_2{1:K}} follow from the independence of the contents on $K$ servers from an $(N',K)$-MDS coded storage system (see \cite[Lemma 3]{sun2017private} for proof details). 

Since \eqref{A_2{1:K}} holds for every $K$-element subset of $[1:N']$, summing the $\binom{N'}{K}$ inequalities and averaging over them yields
\begin{align}
D&\geq PL+\frac{1}{\binom{N'}{K}}
\sum_{\substack{\mathcal{K}\subseteq[1:N']\\|\mathcal{K}|=K}}
H(A_{\mathcal{K}}^{[M-P+1:M]}|W_{1:P}, \mathcal{Q})\notag\\
&\geq PL+\frac{K}{N'}
H(A_{1:N'}^{[M-P+1:M]}|W_{1:P},\mathcal{Q})\label{KHan's}\\
&=PL+\frac{K}{N'}(M-P)L\label{D-K},
\end{align}
where \eqref{KHan's} follows from Han's Inequality and \eqref{D-K} follows from \eqref{(M-P)L}. From \eqref{D-K}, we have 
\begin{align*}
C_{\mathrm{multi\text{-}file\ MDS\text{-}PIR}}(N')
\leq 
\left(
1+
\frac{K(M-P)}
{PN'}
\right)^{-1},
\end{align*}
which proves \eqref{multiMDSPIR}.

\end{document}